\documentclass{article} % For LaTeX2e
\usepackage{iclr2027_conference,times}

\usepackage{amsmath}
\usepackage{amsthm}
\usepackage{amsfonts}
\usepackage{bbm}
\usepackage{graphicx}
\usepackage{subcaption}
\usepackage{tikz}
\usepackage{thmtools} % Prevents cleveref from confusing definitions with theorems in references
\usepackage{enumitem}
\usepackage{thm-restate}
\usepackage{hyperref}
\usepackage{cleveref}

\Crefname{appsec}{Appendix}{Appendices}

\usetikzlibrary{quantikz2}

\setenumerate[1]{label=(\arabic*)} % enumerate labels in parentheses

\newtheorem{theorem}{Theorem}
\newtheorem{proposition}{Proposition}
\newtheorem{lemma}{Lemma}
\newtheorem{corollary}{Corollary}

\newtheorem*{remark}{Remark}

\newtheorem{definition}{Definition}
\newtheorem{hypothesis}{Hypothesis}

\newcommand{\CC}{\mathbb C}

\newcommand{\RR}{\mathbb R}

\newcommand{\cA}{\mathcal{A}}
\newcommand{\cC}{\mathcal{C}}
\newcommand{\cD}{\mathcal{D}}

\newcommand{\cH}{\mathcal{H}}

\newcommand{\cL}{\mathcal{L}}

\newcommand{\cU}{\mathcal{U}}

\newcommand{\ol}{\overline}

\newcommand{\rarr}{\rightarrow}

\newcommand{\bs}[1]{\boldsymbol{#1}}

\newcommand{\parens}[1]{\left(#1\right)}
\newcommand{\brackets}[1]{\left[#1\right]}
\newcommand{\braces}[1]{\left\{#1\right\}}
\newcommand{\angles}[1]{\langle#1\rangle}

\newcommand{\abs}[1]{\left\lvert #1 \right\rvert}
\newcommand{\norm}[1]{\left\lVert #1 \right\rVert}

\newcommand{\twonorm}[1]{\norm{#1}_2}

\newcommand{\ceil}[1]{\left\lceil #1 \right\rceil}

\newcommand*\diff{\mathop{}\!\mathrm{d}}

\newcommand{\op}[2]{|#1\rangle\langle#2|}

\newcommand{\iunit}{\bs{\mathrm{i}}}

\renewcommand{\Re}{\operatorname{Re}}
\renewcommand{\Im}{\operatorname{Im}}

\DeclareMathOperator*{\EE}{\mathbb E}
\DeclareMathOperator{\Haar}{Haar}
\DeclareMathOperator{\Id}{Id}
\DeclareMathOperator{\poly}{poly}
\DeclareMathOperator{\rank}{rank}
\DeclareMathOperator{\Tr}{Tr}

\title{Fragmentation is Efficiently Learnable by Quantum Neural Networks}

\author{Mikhail Mints \\
Department of Computing + Mathematical Sciences \\
California Institute of Technology \\
Pasadena, CA 91125 \\
\texttt{mikhail.mints@student-cs.fr} \\
\AND
Eric R. Anschuetz \\
Institute for Quantum Information and Matter \& Walter Burke Institute for Theoretical Physics \\
California Institute of Technology \\
Pasadena, CA 91125 \\
\texttt{e-ans@nus.edu.sg} \\
}

\iclrfinalcopy % Uncomment for camera-ready version, but NOT for submission.
\begin{document}

\maketitle

\begin{abstract}
In certain classes of physical quantum systems, the exponentially large state space ``fragments'' into many low-dimensional, dynamically disconnected subspaces. We introduce a learning problem known as \emph{fragment classification}, where given a quantum state input, one is interested in classifying to which subspace the state belongs. We prove that solving this learning problem is efficient on a quantum computer when the fragmentation phenomenon satisfies certain conditions. Furthermore, we give evidence supporting the classical hardness of this task by demonstrating that known dequantization techniques fail for the fragment classification problem. Consequently, this work provides a rare example of a physically motivated quantum machine learning task that is both efficient for quantum computers to perform and admits no known classical dequantization.
\end{abstract}

\section{Introduction}

Quantum algorithms are known to outperform the best known classical algorithms on many computational problems exhibiting some sort of algebraic structure~\citep{shor1999polynomial,gyurik2024quantum,berry2024analyzing,harrow2009quantum}. A significant research direction today is focused on trying to understand whether the unique advantages offered by quantum computers can help in machine learning tasks and whether \emph{quantum neural networks} (QNNs) can be superior to classical methods. Classical neural networks are known to behave like Gaussian processes in the asymptotic limit \citep{Neal_1996}, which provides theoretical guarantees on their efficient trainability in a wide range of applications. On the other hand, QNNs are known to have poorly-behaved loss landscapes. They are often dominated by \emph{barren plateaus}~\citep{mcclean2018barren,cerezo2021cost,ragone2023unified}, which generally prevent gradients from being estimated efficiently, and \emph{poor local minima}, which generally prevent gradient descent from reaching the global optimum \citep{Anschuetz_2022_Critical_Points,anschuetz2022barren,you2022convergence}. While there do exist settings where QNNs are efficient to train via gradient descent---such as when they have many symmetries~\citep{equivariant,schatzki2024theoretical,PRXQuantum.5.020328,PhysRevResearch.6.013241,anschuetzgao2022,PRXQuantum.5.030320,anschuetz2024arbitrary}---in these settings there often also exist efficient classical simulation algorithms that render the use of a quantum computer for the task unnecessary outside of potentially an initial data acquisition phase~\citep{Anschuetz_2023_Symmetric,Goh_2025,cerezo2023does}.

Recent work \citep{Anschuetz_2025} formulated the Jordan Algebraic Wishart System (JAWS) framework, which is the first full theoretical characterization of QNN loss landscapes that provides the conditions under which they can become efficiently trainable using gradient descent. In this work, we use these tools to formally demonstrate a specific setting where QNNs \emph{can} in fact be efficiently trained to solve a physically motivated problem for which no efficient classical simulation algorithm is known. Our approach is to showcase a setting where the QNN has a high-dimensional symmetry group but, crucially, the end-user does not know what that group is and---to the best of our knowledge---has no efficient strategy for learning it. At a high level, this mirrors the mechanism underlying other well-known quantum-classical separations, such as the abelian hidden subgroup problem~\citep{kitaev1995quantummeasurementsabelianstabilizer}.

The rest of our paper is organized as follows. First, in \Cref{sec:prelims} we give background on the physical mechanism underlying our learning task known as Hilbert space fragmentation, as well as the general quantum learning technique we will be utilizing here. Second, in \Cref{sec:main_results} we demonstrate our main results: we construct a QNN architecture, show that it satisfies the conditions required for efficient trainability, and then discuss the conjectured classical hardness of the learning task. In \Cref{sec:applications}, we outline a scenario in which our construction could be potentially used in an experimental setting. Finally, we conclude in \Cref{sec:conc}.

\section{Preliminaries}\label{sec:prelims}

\subsection{Fragmentation and the Schur Basis}

The problem that we study in this work is that of classifying states in systems which exhibit \emph{Hilbert space fragmentation}~\citep{PhysRevX.9.021003,PhysRevX.10.011047,PhysRevB.101.174204}, a well-studied barrier to ergodicity in quantum many-body systems conceptually similar to the better-known quantum many-body scarring phenomenon~\citep{bernien2017probing,turner2018weak,Moudgalya_2024}. The mechanism underlying Hilbert space fragmentation has been recently mathematically characterized~\citep{Moudgalya_2022}, and we review this characterization here.

We let $L$ be the problem size---physically known as the \emph{system size}---which we will later take to be large. For each $L$, one is given a set of Hermitian operators $\mathcal{S}$ known as \emph{Hamiltonians} acting on a Hilbert space $\mathcal{H}$ with a tensor product structure $\mathcal{H}=\bigotimes_{i=1}^L\mathcal{H}_i$. $\mathcal{S}$ is constructed such that all Hamiltonians $\bs{H}\in\mathcal{S}$ take the form:
\begin{equation}
    \bs{H}=\sum_{i=1}^m c_i\bs{h}_i,
\end{equation}
where all $c_i\in\mathbb{R}$ and the $\bs{h}_i$ are \emph{local operators}, i.e., $\bs{h}_i$ acts nontrivially only on an $L$-independent number of tensor product factors $\mathcal{H}_i$. Furthermore, $m$ only grows at most polynomially quickly with $L$.

Let $\mathcal{A} = \langle \bs{h}_1, \dots, \bs{h}_m \rangle$ be the associative algebra generated by these operators under addition and matrix multiplication. Let $\cC$ be the \emph{commutant algebra}, consisting of all operators in $\cL(\cH)$ that commute with all elements of $\cA$. By the Von Neumann bicommutant theorem~\citep{Landsman_1998_bicommutant}, we know that $\cA$ and $\cC$ are each other's centralizers, and we can decompose the Hilbert space as:
\begin{equation}
\cH = \bigoplus_{\lambda=1}^{\Lambda} \cH_\lambda^{(\cA)} \otimes \cH_\lambda^{(\cC)},
\end{equation}
where each $\cH_\lambda^{(\cA)}$ is an irreducible representation of $\cA$ and each $\cH_\lambda^{(\cC)}$ is an irreducible representation of $\cC$. Let $N_\lambda = \dim\parens{\cH_\lambda^{(\cA)}}, N'_\lambda = \dim\parens{\cH_\lambda^{(\cC)}}$. In this decomposition, each $\cH_\lambda^{(\cA)}$ is called a \emph{Krylov subspace} of the system, preserved by the action of Hamiltonians constructed from the generators $\bs{h}_1, \dots, \bs{h}_m$. Each of these subspaces has $N'_\lambda$ degenerate copies. The representation in $\cH$ of any operator $A \in \cA$ can be written (in a slight abuse of notation identifying the operator with its representation) in the form:
\begin{equation}
\bs{A} = \bigoplus_{\lambda=1}^{\Lambda} \bs{A}^\lambda \otimes \Id_{N'_\lambda}.
\end{equation}
That is, any such operator acts as the identity on the ``multiplicity labels'' in $\cH_\lambda^{(\cC)}$, meaning that it acts identically on each of the $N_\lambda'$ degenerate Krylov subspaces isomorphic to $\cH_\lambda^{(\cA)}$.

The phenomenon of \emph{Hilbert space fragmentation} is said to occur when the total multiplicity, $\sum_\lambda N_\lambda'$, is exponential in the system size $L$ \citep{Moudgalya_2022}. We are interested in the stronger condition where the dimension of $\cA$ is polynomial in the system size. If a system does not have this property, we may be able to restrict to a subspace that does: we pick $\Lambda'$ as a function of $L$ to be such that
\begin{equation}
N = \sum_{\lambda=1}^{\Lambda'} N_\lambda = \poly(L).
\end{equation}
Note that this $\Lambda'$ may be a constant or it may be equal to $\Lambda$ depending on the system in question and the degree of degeneracy in the Krylov subspaces. We additionally make an assumption that $N_\lambda = \omega(1)$ for all $\lambda \in [\Lambda']$ - that is, the dimension of all sectors must grow as a function of $L$. Now, we can define the space
\begin{equation}
\cH' = \bigoplus_{\lambda=1}^{\Lambda'} \cH_\lambda^{(\cA)} \otimes \cH_\lambda^{(\cC)}.
\end{equation}
The representation in $\cH'$ of any operator $A \in \cA$ can then be written as
\begin{equation}
\bs{A} = \bigoplus_{\lambda=1}^{\Lambda'} \bs{A}^\lambda \otimes \Id_{N'_\lambda}.
\end{equation}

While this construction may seem ad hoc and technical, surprisingly the phenomenon of Hilbert space fragmentation occurs in many quantum systems. It underlies the ergodicity breaking of the Temperley--Lieb model~\citep{Batchelor_1990}, the $t$--$J_z$ model~\citep{Zhang_1997_T_Jz}, and the Pair-Flip model~\citep{Moudgalya_2022}.

One can also define an orthonormal basis of $\mathcal{H}$ respecting the decomposition of $\mathcal{H}$ into Krylov subspaces~\citep{Moudgalya_2022}. In analogy with the terminology for permutation-invariant quantum systems~\citep{PhysRevLett.97.170502}, we can construct a \emph{Schur basis} for $\cH'$ as a set of states $\ket{\lambda, q_\lambda, p_\lambda} = \ket{\lambda, q_\lambda} \otimes \ket{\lambda, p_\lambda}$ such that the states $\ket{\lambda, q_\lambda}$ form an orthonormal basis for $\cH_\lambda^{(\cA)}$ and the states $\ket{\lambda, p_\lambda}$ form an orthonormal basis for $\cH_\lambda^{(\cC)}$.

\subsection{The Classification Task}

Given a quantum state that is a basis state promised to be in the Schur basis, we would like to perform some kind of quantum measurement to determine which state exactly it is, ignoring the degeneracy. In some special settings---such as when the system is permutation invariant, and one is interested in the associated Schur basis of that Krylov decomposition---this can be done efficiently using a quantum algorithm known as the Schur transform~\citep{PhysRevLett.97.170502}. However, generally, we know only the generators $\bs{h}_1, \dots, \bs{h}_m$. Instead, we hope in effect to \emph{learn} the Schur transform associated with $\mathcal{S}$, assuming we are also given copies of Schur basis states $\ket{\lambda, q_\lambda, p_\lambda}$. Specifically, we want to learn a transformation $\bs{U}$ such that
\begin{equation}
\Tr_{\cH'} \brackets{\bs{U}^\dagger (\op{\lambda, q_\lambda, p_\lambda}{\lambda, q_\lambda, p_\lambda} \otimes \op{\bs{0}}{\bs{0}}) \bs{U}} \approx \op{\bs{\lambda}, \bs{q_\lambda}}{\bs{\lambda}, \bs{q_\lambda}}.
\end{equation}
Here, we have added an ancillary register of $n_a$ qubits with Hilbert space of dimension $N_a = 2^{n_a}$ and initialized it in the computational basis state $\ket{\bs{0}}$. We also use the notation $\ket{\bs{\lambda},\bs{q_\lambda}}$ to label bit strings encoded in the computational basis labeling $\lambda$ and $q_\lambda$. The unitary $U$ will be implemented as a quantum circuit that should write a bit string label $\bs{\lambda}, \bs{q_\lambda}$ to this ancillary register in the computational basis. Note that here we are ambivalent to the multiplicity label $p_\lambda$; there are exponentially many such labels, and we cannot hope to learn them efficiently. Instead, we also require that our network \emph{generalizes} its classification of any $\ket{\lambda, q_\lambda, p_\lambda}$ given training examples of $\ket{\lambda, q_\lambda, p_\lambda'}$ for $p_\lambda$ not necessarily equal to $p_\lambda'$. We call this problem \emph{fragment classification}.

\begin{definition}[Fragment Classification]\label{def:fragment_classification}
Suppose that we have a dataset
\begin{equation}
\cD = \{(\ket{x}, \bs{O}_x)\}_{x \in [M]},
\end{equation}
of size $M \leq N$ where each $\ket{x}$ is of the form $\ket{\lambda, q_\lambda, p_\lambda} \otimes \ket{0}$, and the corresponding $\bs{O}_x$ is of the form
\begin{equation}
\bs{O}_x = -\Id_{\cH'} \otimes \op{\bs{\lambda}, \bs{q_\lambda}}{\bs{\lambda}, \bs{q_\lambda}},
\end{equation}
where $\bs{\lambda}, \bs{q_\lambda}$ are bit string representations of $\lambda, q_\lambda$ on the ancillary register of $n_a = \ceil{\log_2 M}$ qubits. Let $N_a = 2^{n_a}$ be the dimension of the Hilbert space of the ancillary register. We say that a quantum learning algorithm solves the fragment classification problem if, given the dataset $\cD$, for any $\epsilon > 0$ it can construct in $\poly(L, 1/\epsilon)$ time a quantum circuit representing a unitary $\bs{U}$ acting on $\cH' \otimes \CC^{N_a}$, such that for any Schur basis state $\ket{\lambda, q_\lambda, p_\lambda}$ (including those $p_\lambda$ for which the state is not part of $\cD$), we have
\begin{equation}
\bra{\bs{\lambda}, \bs{q_\lambda}} \Tr_{\cH'} \brackets{\bs{U}^\dagger (\op{\lambda, q_\lambda, p_\lambda}{\lambda, q_\lambda, p_\lambda} \otimes \op{\bs{0}}{\bs{0}}) \bs{U}} \ket{\bs{\lambda}, \bs{q_\lambda}} \geq 1 - \epsilon.
\end{equation}
\end{definition}

\begin{remark}
In a typical situation where fragment classification might be used, it could be difficult to obtain a dataset containing multiple orthogonal Schur basis states from a given sector $\lambda$ labeled with $(\lambda, q_\lambda)$. A more realistic scenario is one in which only $\lambda$ is learned. This can be done in our formulation of the problem by simply removing the $\bs{q_\lambda}$ from the bitstring representation. Note that in this case, the model will still generalize across different $p_\lambda$ (but not across different $q_\lambda$).
\end{remark}

\subsection{Variational Quantum Algorithms}

The quantum learning algorithm we will later construct will be a \emph{variational quantum algorithm}~\citep{peruzzo}. Here, one hopes to learn a unitary $\bs{U}$ by tuning parameters of a variational ansatz constructed using exponentials of local operators. The variational parameters of the ansatz are then tuned using a gradient-based optimizer to minimize some loss function.

Specializing to our setting, $\bs{U}\left(\bs{\theta}\right)$ will be constructed out of matrix exponentials of Hamiltonians constructed from the given generators $\bs{h}_1, \dots, \bs{h}_m$, and is parameterized by $\bs{\theta} \in \RR^p$. We use gradient descent to find optimal values of the parameters to minimize the loss function:
\begin{equation}
\ell(\bs{\theta}) =
\frac{1}{M} \sum_{x=1}^M \bra{x} \bs{U}(\bs{\theta}) \bs{O}_x \bs{U}(\bs{\theta})^\dagger \ket{x},
\end{equation}
where each $\ket{x}$ is of the form $\ket{\lambda, q_\lambda, p_\lambda} \otimes \ket{\bs{0}}$ and the corresponding \emph{objective observable} is
\begin{equation}
\bs{O}_x = -\Id_{\cH'} \otimes \op{\bs{\lambda}, \bs{q_\lambda}}{\bs{\lambda}, \bs{q_\lambda}}.
\end{equation}
The training process is illustrated in \Cref{fig:training_diagram}. We will give full details of the construction of $\bs{U}(\bs{\theta})$ in \Cref{sec:ansatz_const}, and demonstrate that $\ell\left(\bs{\theta}\right)$ can be efficiently optimized in \Cref{sec:eff_train}.

\begin{figure}
    \centering
    \includegraphics[scale=0.75]{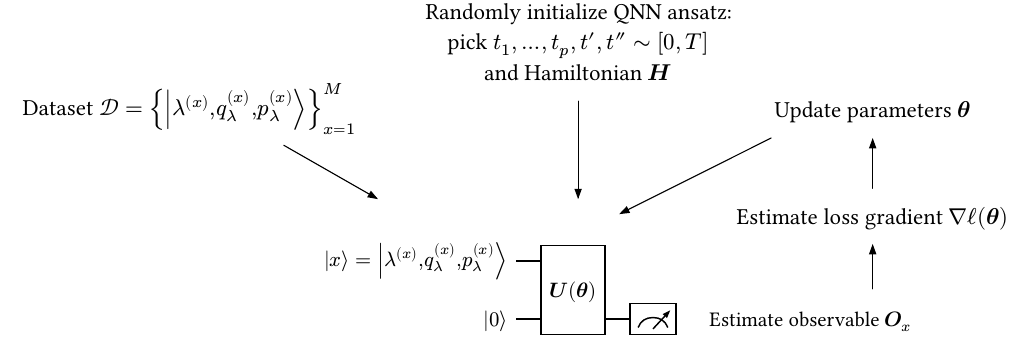}
    \caption{Diagram of our QNN training process. We are given a dataset of Schur basis states and randomly sample a QNN ansatz architecture and an initial vector of parameters (see \Cref{sec:ansatz_const}). We then compile the QNN into a quantum circuit and repeatedly run on the input states to estimate the gradient of the loss function, which is then used to adjust the parameters.}
    \label{fig:training_diagram}
\end{figure}

\subsection{Trainability Conditions}

To formally prove that QNNs can efficiently solve this problem, we need to demonstrate the two main conditions needed for a QNN to be efficiently trainable. The following are standard criteria of QNN trainability in the literature~\citep{Anschuetz_2025,cerezo2023does}:

\begin{enumerate}
\item \textbf{The absence of barren plateaus}. In the asymptotic limit of large system size, the distribution of the loss landscape over the random initialization of the QNN architecture converges to one where the derivatives only decay polynomially with the system size, thus allowing the gradient to be efficiently estimated on a quantum computer in polynomial time. We formally prove this in \Cref{thm:gradient_calculation}.
\item \textbf{The absence of poor local minima}. When the number of parameters $p$ is sufficiently large, the loss landscape enters an \emph{overparameterized regime} where all spurious local minima disappear. We follow the approach of \citet{Larocca_2023} and show the Hessian is not full-rank in order to demonstrate a lack of poor local minima. The value of $p$ needed for this to occur is only polynomially large in the system size, making the QNN in the overparameterized phase efficiently implementable on a quantum computer. We formally prove this in \Cref{prop:hessian_calculation}. We then apply the Kac-Rice formula~\citep{Adler_2007} to the loss function modified by a small perturbation, demonstrating the absence of poor local minima in \Cref{thm:local_minima}.
\end{enumerate}

\section{Main Results}\label{sec:main_results}

\subsection{The QNN Ansatz Construction}\label{sec:ansatz_const}
Suppose that we have a dataset
\begin{equation}
\cD = \{(\ket{x}, \bs{O}_x)\}_{x \in [M]},
\end{equation}
of size $M \leq N$ as in \Cref{def:fragment_classification}. We assume we are given at most one example $x$ for any given $\lambda,q_\lambda$ pair - we will see in \Cref{prop:generalizability} that we do not require an example for each $p_\lambda$. Let $N_a = 2^{n_a}$ be the dimension of the Hilbert space of the ancillary register. Let
\begin{equation}
\cH^\ast = \cH' \otimes \CC^{N_a} = \bigoplus_{\lambda=1}^{\Lambda'} \cH_\lambda^\ast \otimes \cH_\lambda^{(\cC)},
\end{equation}
where we define $\cH_\lambda^\ast = \cH_\lambda^{(\cA)} \otimes \CC^{N_a}$, of dimension $N_\lambda^\ast = N_\lambda N_a$.

Now, we construct a randomized QNN ansatz parameterized by $\bs{\theta} = (\theta_1, \dots, \theta_p) \in [0, 2\pi)^p$ as follows:
\begin{equation}\label{eq:ansatz}
\bs{U}(\bs{\theta}) = e^{\iunit \bs{H} t''} \parens{\prod_{i=1}^p e^{-\iunit \bs{H} t_i} e^{\iunit \bs{A} \theta_i} e^{\iunit \bs{H} t_i}} e^{-\iunit \bs{H} t'},
\end{equation}
where we select $t_1, \dots, t_p, t', t''$ i.i.d. uniformly from $[0, T]$ for some fixed $T$, we let $\bs{H}$ be a Hamiltonian on the system and the ancillary register:
\begin{equation}\label{eq:hamiltonian}
\bs{H} = \sum_{i=1}^m c_i \bs{h}_i \otimes \parens{\sum_{j=1}^{n_a} c'_{i,j,x} \sigma^x_j + c'_{i,j,y} \sigma^y_j + c'_{i,j,z} \sigma^z_j},
\end{equation}
and we let $\bs{A}$ be a local operator block-diagonal in the Schur basis, which we assume without loss of generality is positive semidefinite to simplify our analysis---we can just take it to be a large multiple of the identity added to $\bs{h}_1 \otimes \Id_{N_a}$.

We assume that coefficients in \Cref{eq:hamiltonian} are chosen such that the Hamiltonian obeys the \emph{full Eigenstate Thermalization Hypothesis} (ETH)~\citep{PhysRevE.99.042139} in each Krylov subspace. The full ETH is an ansatz for higher-order time-averaged correlation functions of local observables, and has been shown to imply a quantitative scrambling behavior~\citep{Fava_2025}. More formally, we assume the following condition holds in each Krylov subspace:
\begin{hypothesis}[Full ETH, free cumulant definition, from~{\citet[Eq.~(5)]{Fava_2025}}]\label{hyp:full_eth}
    We say a sequence of Hamiltonians $\bs{H}$ describing a size-$L$ system and acting on an $N$-dimensional Hilbert space obeys the full ETH if the following is true. For any $k$, there exists some $t_k = O\left(\poly(L,k)\right)$ such that choosing $t$ uniformly from $\left[0,t_k\right]$ and defining the free cumulant $\kappa_{2k}$ with respect to this distribution, for any local operators $\bs{A}_i\left(t\right)$ time-evolved under $\bs{H}$ and any local operators $\bs{B}_i$:
    \begin{equation}
        \kappa_{2k}\left[\bs{A}_1\left(t\right),\bs{B}_1,\bs{A}_2\left(t\right),\bs{B}_2,\ldots,\bs{A}_k\left(t\right),\bs{B}_k\right]=O\left(1/N\right).
    \end{equation}
\end{hypothesis}

\begin{remark}
It is widely believed that the full ETH generally holds in physical systems in the absence of symmetries. While we make this assumption to make the calculations for the gradient distribution more tractable, we expect the ETH condition to be a ``worst case'' for barren plateaus: if the variational ansatz does not uniformly mix over the accessible Hilbert space and is biased towards a smaller subspace, we would expect the gradients of the loss landscape to vanish less quickly.
\end{remark}

Observe that we can block-diagonalize all operators in \Cref{eq:ansatz}, writing them as a sum over the irreducible representations (algebraic sectors) of $\cA$. We identify the operators with their representations in $\cH^*$:
\begin{equation}
\bs{U}(\bs{\theta}) = \bigoplus_{\lambda=1}^{\Lambda'} \bs{U}^\lambda(\bs{\theta}),
\end{equation}
where
\begin{equation}
\bs{U}^\lambda(\bs{\theta}) = e^{\iunit \bs{H}^\lambda t''} \parens{\prod_{i=1}^p e^{-\iunit \bs{H}^\lambda t_i} e^{\iunit \bs{A}^\lambda \theta_i} e^{\iunit \bs{H}^\lambda t_i}} e^{-\iunit \bs{H}^\lambda t'}.
\end{equation}

Define the loss function
\begin{equation}
\ell(\bs{\theta}) = \frac{1}{M} \sum_{x=1}^M \ell_x (\bs{\theta}) =
\frac{1}{M} \sum_{x=1}^M \bra{x} \bs{U}(\bs{\theta}) \bs{O}_x \bs{U}(\bs{\theta})^\dagger \ket{x}.
\end{equation}
We can also decompose this by sector. Note that since each input state is restricted to one of the sectors, we can write
\begin{equation}
\ell(\bs{\theta}) = \sum_{\lambda=1}^{\Lambda'} \sum_{x=1}^{M_\lambda} \ell_x^\lambda (\bs{\theta}) =
\frac{1}{M} \sum_{\lambda=1}^{\Lambda'} \sum_{x=1}^{M_\lambda} \bra{x} \bs{U}^\lambda(\bs{\theta}) \bs{O}^\lambda_x \bs{U}^\lambda(\bs{\theta})^\dagger \ket{x}.
\end{equation}
Note that the indexing of $x$ is separate for each $\lambda$, so we have slightly abused notation here and are letting each $x$ in the above equation corresponds to some $q_\lambda$, where $M_\lambda$ is the number of dataset entries in the $\lambda$ block.

We initially randomly choose the QNN ansatz by sampling $t_1, \dots, t_p, t', t''$ and choosing the Hamiltonian $\bs{H}$. Then, these become fixed, and only the parameterized parts of the form $e^{\iunit \bs{A} \theta_i}$ will be changed during training as we adjust the parameters $\theta_i$. To analyze the gradient of the loss function, we will consider its distribution at any particular value of $\bs{\theta}$ over the random choice of the QNN ansatz. We perform a sequence of reductions, first reducing the distribution to the case where we only have to consider $\bs{\theta} = 0$ and then writing the loss derivatives in terms of Gaussian-distributed random variables in the asymptotic limit, allowing us to prove the result about the variance of the gradient.

\subsection{Efficient Trainability}\label{sec:eff_train}

First, we establish the prior claim that indeed the training of our QNN ansatz does not require a data set of Schur basis states over all multiplicity labels.

\begin{proposition}\label{prop:generalizability}
If we take the (exponential-size) dataset $\mathcal{D}'$ consisting of every Schur basis state $\ket{\lambda,q_\lambda,p_\lambda}$ for $\lambda,q_\lambda$ present in $\mathcal{D}$, the loss $\ell_{\mathcal{D}'}(\bs{\theta})$ with respect to that dataset is upper-bounded as $\ell_{\mathcal{D}'}(\bs{\theta}) \leq \max_{\lambda,x} \ell_x^\lambda (\bs{\theta})$.
\end{proposition}

\begin{proof}
This follows from the fact that every layer of the QNN is constructed from matrix exponentials of elements of $\mathcal{A}$. Thus, conjugating by $\bs{U}(\bs{\theta})$ will act in the same way on each degenerate Krylov subspace, which means that the loss contribution from the new datapoints for a given $\lambda, q_\lambda$ will be exactly the same as that for the point in the original dataset. So, we can write
\begin{equation}
\ell_{\mathcal{D}'}(\bs{\theta}) = \frac{1}{\abs{\mathcal{D}'}} \sum_{\lambda=1}^{\Lambda'} N'_\lambda \sum_{x=1}^{M_\lambda} \ell_x^\lambda (\bs{\theta}) \leq
\frac{\max_{\lambda,x} \ell_x^\lambda (\bs{\theta})}{\abs{\mathcal{D}'}} \sum_{\lambda=1}^{\Lambda'} N'_\lambda M_\lambda \leq \max_{\lambda,x} \ell_x^\lambda (\bs{\theta}).
\end{equation}
\end{proof}

% Thus, if we can efficiently train our QNN on the polynomially-sized dataset $\mathcal{D}$, we know that it will correctly classify any Schur basis state. Note that if the training data points contribute equally to the loss, the loss on $\mathcal{D'}$ will be the same as that on $\mathcal{D}$. Thus, in this setting, we do not need to worry about the typical problem of \emph{overfitting} faced by many machine learning models, as the algebraic structure of the ansatz prevents it from occurring.

Now, we want to analyze the derivatives of this loss function.
\begin{restatable}{lemma}{restateFirstDerivAtZero}\label{lem:1st_deriv_at_0}
At $\bs{\theta} = 0$, we have that
\begin{equation}
\partial_i \ell(\bs{\theta}) \Big|_{\bs{\theta} = 0} =
\frac{\iunit}{M} \sum_{\lambda=1}^{\Lambda'} \sum_{x=1}^{M_\lambda}
\bra{x} e^{\iunit \bs{H}^\lambda t''} \brackets{e^{-\iunit \bs{H}^\lambda t_i} \bs{A}^\lambda e^{\iunit \bs{H}^\lambda t_i}, \; e^{-\iunit \bs{H}^\lambda t'} \bs{O}^\lambda_x e^{\iunit \bs{H}^\lambda t'}} e^{-\iunit \bs{H}^\lambda t''} \ket{x}.
\end{equation}
\end{restatable}

Now, we want to analyze the distribution of the gradient of the loss function. To do this, we consider joint distributions of individual terms that contribute to the expression for the gradient.

Using \Cref{lem:eth_moment_matching_gradient}, we can perform a reduction where we approximate the distribution of the first derivatives by taking each $\bs{g}_i^\lambda, \bs{g}'^\lambda, \bs{g}''^\lambda$ to be Haar-random.

\begin{restatable}{lemma}{restateGradientHaarReduction}\label{lem:gradient_haar_reduction}
We can approximate, at any value of $\bs{\theta}$,
\begin{equation}
\partial_i \ell_x^\lambda(\bs{\theta}) \rightsquigarrow
\frac{\iunit}{M} \bra{x} \bs{g}''^{\lambda \dagger} \brackets{\bs{g}_i^\lambda \bs{A}^\lambda \bs{g}_i^{\lambda \dagger}, \; \bs{g}'^\lambda \bs{O}^\lambda_x \bs{g}'^{\lambda \dagger}} \bs{g}''^{\lambda} \ket{x},
\end{equation}
where $\bs{g}_1, \dots, \bs{g}_p, \bs{g'}, \bs{g''} \sim \Haar$ up to an $o(1)$ error in the L\'evy-Prokhorov metric (\Cref{def:levy_prokhorov}), if we choose a sufficiently large $T$ and we assume that $N_\mathrm{min} = \min_\lambda N_\lambda = \omega(1)$. Furthermore, this approximation is valid for the joint distribution of any \emph{two} gradient contributions $\partial_i \ell_x^\lambda(\bs{\theta})$ and $\partial_j \ell_y^{\lambda'}(\bs{\theta})$: their joint distribution can be approximated up to an $o(1)$ L\'evy-Prokhorov error by replacing the Hamiltonian time evolution with Haar-random matrices.
\end{restatable}

\begin{restatable}{lemma}{restateGradientGaussianReduction}\label{lem:gradient_gaussian_reduction}
We can approximate, at any value of $\bs{\theta}$,
\begin{equation}
\partial_i \ell_x^\lambda (\bs{\theta}) \rightsquigarrow
\frac{\iunit}{M {N_\lambda^\ast}^2}
\bra{x} \bs{g}''^{\lambda \dagger} \brackets{\widetilde{\bs{g}}_i^\lambda \bs{A}^\lambda \widetilde{\bs{g}}_i^{\lambda \dagger}, \; \widetilde{\bs{g}}'^\lambda \bs{O}^\lambda_x \widetilde{\bs{g}}'^{\lambda \dagger}} \bs{g}''^{\lambda} \ket{x},
\end{equation}
where $\widetilde{\bs{g}}_1^\lambda, \dots, \widetilde{\bs{g}}_p^\lambda, \widetilde{\bs{g}}'^\lambda$ have i.i.d. standard Gaussian entries, up to an $o(1)$ error in the L\'evy-Prokhorov metric, and the same approximation is valid for the joint distribution of any two loss contributions $\partial_i \ell_x^\lambda (\bs{\theta})$ and $\partial_j \ell_y^{\lambda'} (\bs{\theta})$.
\end{restatable}

We know by construction that $\bs{A}$ is positive semidefinite. We can diagonalize it, writing each block as
\begin{equation}
\bs{A}^\lambda = \sum_{\mu=1}^{N_\lambda^\ast} a^\lambda_\mu \op{\mu}{\mu}
\end{equation}
Now, we want to perform an additional reduction to replace $\bs{A}$ with a \emph{semi-isotropic} $\widetilde{\bs{A}}$, which we define as:
\begin{equation}
\widetilde{\bs{A}}^\lambda = \frac{\Tr(\bs{A}^\lambda)}{r_A^\lambda} \sum_{\mu=1}^{r_A^\lambda} \op{\mu}{\mu},
\end{equation}
where
\begin{equation}
r_A^\lambda = \ceil{\frac{\Tr(\bs{A}^\lambda)^2}{\Tr((\bs{A}^\lambda)^2)}}.
\end{equation}
That is, we want to replace $\bs{A}^\lambda$ with a matrix all of whose nonzero eigenvalues are the same, preserving the average eigenvalue of $\bs{A}^\lambda$.

Observe that
\begin{equation}
\Tr\parens{\widetilde{\bs{A}}^\lambda} = \Tr\parens{\bs{A}^\lambda}
\end{equation}
and (up to a small incurred by rounding $r_A^\lambda$ to an integer),
\begin{equation}
\Tr\parens{(\widetilde{\bs{A}}^\lambda)^2} = \frac{\Tr(\bs{A}^\lambda)^2}{r_A^\lambda} = \Tr((\bs{A}^\lambda)^2).
\end{equation}

\begin{restatable}{lemma}{restateSemiIsotropicAReduction}\label{lem:semi_isotropic_A_reduction}
Replacing $\bs{A}$ with $\widetilde{\bs{A}}$, we can write
\begin{equation}
\partial_i \ell_x^\lambda (\bs{\theta}) \rightsquigarrow
\hat{\ell}^\lambda_{x; i} := \frac{\iunit}{M {N_\lambda^\ast}^2}
\bra{x} \bs{g}''^{\lambda \dagger} \brackets{\widetilde{\bs{g}}_i^\lambda \widetilde{\bs{A}}^\lambda \widetilde{\bs{g}}_i^{\lambda \dagger}, \; \widetilde{\bs{g}}'^\lambda \bs{O}^\lambda_x \widetilde{\bs{g}}'^{\lambda \dagger}} \bs{g}''^{\lambda} \ket{x},
\end{equation}
up to an $o(1)$ error in the L\'evy-Prokhorov metric, assuming that $N_\mathrm{min} = \min_\lambda N_\lambda = \omega(1)$. The same is true for joint distributions of two derivative contributions as before.
\end{restatable}

Now that we have performed these reductions, we can, up to a vanishing error in L\'evy-Prokhorov metric, approximate the distribution of the gradient of the loss function as some $\widehat{\nabla \ell}$ for which we can find a variance lower bound.

\begin{restatable}{theorem}{restateGradientCalculation}\label{thm:gradient_calculation}
At any value of $\bs{\theta} \in [0, 2\pi)^p$,
\begin{equation}
\EE\brackets{\norm{\widehat{\nabla \ell}}^2} = \sum_{\lambda=1}^{\Lambda'} \frac{2 M_\lambda (N_a - 1) \Tr((\bs{A}^\lambda)^2) p}{M^2 N_a^4 N_\lambda^2} = \Omega\parens{\frac{1}{\poly(L)}},
\end{equation}
where
\begin{equation}
\parens{\widehat{\nabla \ell}}_i \coloneq \sum_{\lambda=1}^{\Lambda'} \sum_{x=1}^{M_\lambda} \hat{\ell}^\lambda_{x; i}.
\end{equation}
\end{restatable}

By proving \Cref{thm:gradient_calculation}, we have formally established one of the two main claims: the absence of barren plateaus. Since the gradient of the loss function only scales inversely polynomially with $L$ and the quantum circuit for our QNN ansatz is polynomial-size, it is efficient to estimate gradients on a quantum computer by performing repeated runs and measurements with varied parameters, and thus gradient descent can be efficiently performed.

We now argue that the model enters the overparameterized regime, as in \citet{Larocca_2023}, by showing that the Hessian of the loss landscape stops being full-rank when the number of parameters $p$ is only polynomial in the system size.

\begin{restatable}{proposition}{restateHessianCalculation}\label{prop:hessian_calculation}
For some $p = O(\poly(L))$, the Hessian matrix $\nabla^2 \ell$ at any value of $\bs{\theta} \in [0, 2\pi)^p$ (where $\bra{i} \parens{\nabla^2 \ell} \ket{j} = \partial_i \partial_j \ell(\bs{\theta})$) does not have full rank.
\end{restatable}

This rank deficiency allows us to constrain the structure of local optima through the use of the Kac--Rice formula~\citep{Adler_2007}, as was done in previous work on machine learning loss landscapes~\citep{pmlr-v38-choromanska15,Anschuetz_2022_Critical_Points,anschuetz2022barren,Anschuetz_2025}. Informally, the Kac--Rice formula relates the joint distribution of the gradient and the Hessian to the expected density of local optima of a random field when that random field is a Morse function with high probability; when the expected Hessian vanishes, there are no spurious local optima.

To make this connection more precise, we consider the joint distribution of $\left(\widehat{\nabla \ell},\nabla^2 \ell\right)$, where $\widehat{\nabla \ell}$ is as we computed in \Cref{lem:semi_isotropic_A_reduction}. As $\ell$ is not a Morse function w.h.p.---a condition necessary for the Kac--Rice formula to hold---we add an order-$\epsilon$ perturbation to the loss function and consider the scaling behavior of the density as $\epsilon\to 0^+$. We show in this limit that the density of local minima above the global minimum vanishes. Thus, we have the following:

\begin{theorem}[No spurious local minima, informal version of \Cref{thm:formal_loc_min}]\label{thm:local_minima}
    For some $p=O(\poly(L))$, the expected number of local minima at loss value greater than the value of the global minimum is zero.
\end{theorem}

The formal proofs and statements of the results in this section can be found in \Cref{app:formal_proofs}.

\subsection{Conjectured Classical Hardness}

There exist many classical algorithms for efficiently learning properties of physical systems in the presence of symmetries~\citep{Anschuetz_2023_Symmetric,Goh_2025,cerezo2023does}. Fundamentally, the strategies of these algorithms make use of the fact that, when the algebraic structure of a learning problem is known a priori, classical algorithms can leverage this knowledge to circumvent the exponential scaling of naive classical simulation methods.

Our setting differs from these prior directions in one important way: we have no prior knowledge of the algebraic structure of the system (indeed, this is the property we are hoping to learn). While we have no formal reduction connecting this problem to standard complexity theoretic conjectures, at a high level, this resembles other problems conjectured to be difficult for classical algorithms, including the hidden subgroup problem~\citep{kitaev1995quantummeasurementsabelianstabilizer} and the state hidden subgroup problem~\citep{10.1145/3717823.3718118}.

The classical algorithms in the line of work mentioned above take as input the structure coefficients $c_{i,j}^k$ of the system's Lie or Jordan algebra, where the action of the algebra's operation $\circ$ (i.e. a commutator or anticommutator) on a basis $\{b_i\}_i$ of the algebra is given by
\begin{equation}
b_i \circ b_j = \sum_{k} c_{i,j}^k b_k.
\end{equation}
In our setting, we are given the algebra's generators $\{\bs{h}_i\}_i$. To obtain a basis and compute the structure coefficients, one would need to take nested group operations of the form $\bs{h}_{i_1} \circ \parens{\bs{h}_{i_2} \circ \parens{\dots \circ \bs{h}_{i_l}} \dots}$, where the depth may be as large as the dimension of the algebra. If we have some sparse classical representation of the generators (for instance, Pauli strings), we cannot guarantee that these nested group operations will have such an efficient representation: in general, their size will blow up exponentially in the depth.

As we have demonstrated that this problem can be efficiently solved using quantum neural networks, this establishes one of the few known settings (outside of sampling-based tasks~\citep{10.1145/3618260.3649722}) where a quantum machine learning task has no known dequantization.

\section{Applications}\label{sec:applications}

The problem we presented as a showcase for our QNN construction is physically motivated, and we hope that it may be of practical interest. Recent experimental work has demonstrated examples of Hilbert space fragmentation~\citep{Bluvstein_2021,Adler_2024}, where the fragmentation phenomenon was not studied algebraically but rather by observing the dynamics of different states when performing a quench. For instance, \citet{Adler_2024} observes different relaxation dynamics between ``dimer'' and ``chequerboard'' states, which belong to different Krylov sectors $\lambda$. The authors perform Hamiltonian evolution under different times in order to generate a collection of states in each of the sectors. Such an experimental setup could be used to create a dataset for training our QNN model.

Now, suppose that some other experiment prepares some unknown state on this system and one wishes to classify which sector it lies in. It should be possible to perform reverse time evolution $e^{\iunit H t}$ on this state to obtain some known initial state (e.g. the ``dimer'' or ``chequerboard'' states), but the time $t$ is not known. Furthermore, the state may be in a subspace corresponding to a different label $p_\lambda$, so it could be difficult to perform the correct measurement to extract the necessary information from the state without a characterization of the system's algebraic structure. But if the state is given to the trained QNN model, it would perform the classification correctly as long as it is provided with training data from the sectors of interest.

\section{Conclusion}\label{sec:conc}

In this work, we demonstrated a physically motivated problem setting where (1) quantum neural networks efficiently solve the problem under a realistic set of assumptions and (2) there are no known existing classical algorithms for simulating these networks, despite the substantial recent progress in methods for dequantizing quantum learning algorithms~\citep{Anschuetz_2023_Symmetric,Goh_2025,cerezo2023does}. At a high level, our results followed by finding a setting where the quantum network was able to take advantage of algebraic structure present in a problem, but for which this structure is not given a convenient classical description for a classical simulator to take advantage of. We invite others to test our results by attempting to develop classical simulation methods in this more restricted setting, as empirical tests of the tractability of this task have implications for not only the potential utility of the algorithm presented here, but also more generally for determining unknown algebraic properties of symmetric quantum systems.

While we examine one specific setting with this property, we believe this strategy for finding physically motivated problems for quantum learning algorithms can be more generally leveraged. We hope this work motivates future studies of settings where quantum neural networks can take advantage of symmetries that are inaccessible to classical computers.

\subsection*{AI use statement}

In this work, we used generative AI tools to help with understanding relevant technical concepts, to survey related work on QNN loss landscapes, and assist with certain routine calculations and steps such as some of the bounds on the L\'evy-Prokhorov metric in \Cref{app:lp_details}. We have checked the correctness of these calculations. AI tools were \emph{not} used to generate proofs, develop theoretical models, design research methodologies, or write the text of the paper. The code for the numerics in \Cref{app:numerics} was written without AI assistance, apart from some queries about the usage of library functions. We take responsibility for the final content of this work, including text, claims or artifacts produced with the aid of generative AI.

% \subsection*{Reproducibility statement}

% We provide formal proofs of our main results in \Cref{app:formal_proofs}, with additional details on error bounds in \Cref{app:lp_details}. We also provide small-scale numerical experiments in \Cref{app:numerics}. The code used to run these experiments is available at \url{https://anonymous.4open.science/r/fragmentation-qnn-anonymized-8415/qnn.ipynb}.

\subsection*{Acknowledgments}
This work was done while M.M.\ was a part of Caltech’s SURF program at the Institute for Quantum Information and Matter. M.M.\ was supported by the Arthur R.\ Adams SURF fellowship and would like to thank John Preskill for his support as a mentor. E.R.A.\ is funded in part by the Walter Burke Institute for Theoretical Physics at Caltech.

\bibliography{references}
\bibliographystyle{iclr2027_conference}

\newpage

\appendix

\crefalias{section}{appsec}

\section{Formal Proofs of the Main Results}\label{app:formal_proofs}

\restateFirstDerivAtZero*

\begin{proof}
\begin{equation}
\begin{aligned}
& \frac{\partial}{\partial \theta_i} \ell(\bs{\theta}) =
\frac{\iunit}{M} \sum_{\lambda=1}^{\Lambda'} \sum_{x=1}^{M_\lambda}
\frac{\partial}{\partial \theta_i} \bra{x} \bs{U}^\lambda(\bs{\theta}) \bs{O}^\lambda_x \bs{U}^\lambda(\bs{\theta})^\dagger \ket{x} = \\
={}& \frac{\iunit}{M} \sum_{\lambda=1}^{\Lambda'} \sum_{x=1}^{M_\lambda}
\frac{\partial}{\partial \theta_i} \bra{x}
e^{\iunit \bs{H}^\lambda t''}
\parens{\prod_{j=1}^p e^{-\iunit \bs{H}^\lambda t_j} e^{\iunit \theta_j \bs{A}^\lambda} e^{\iunit \bs{H}^\lambda t_j}}
e^{-\iunit \bs{H}^\lambda t'} \bs{O}^\lambda_x e^{\iunit \bs{H}^\lambda t'} \\
& \hspace{11em} \parens{\prod_{j=p}^1 e^{-\iunit \bs{H}^\lambda t_j} e^{-\iunit \theta_j\bs{A}^\lambda} e^{\iunit \bs{H}^\lambda t_j}}
e^{-\iunit \bs{H}^\lambda t''} \ket{x} = \\
={}& \frac{\iunit}{M} \sum_{\lambda=1}^{\Lambda'} \sum_{x=1}^{M_\lambda}
\bra{x} \Bigg( e^{\iunit \bs{H}^\lambda t''} \parens{\prod_{j=1}^p e^{-\iunit \bs{H}^\lambda t_j} e^{\iunit \theta_j \bs{A}^\lambda} \parens{\bs{A}^\lambda}^{\delta_{ij}} e^{\iunit \bs{H}^\lambda t_j}}
e^{-\iunit \bs{H}^\lambda t'} \bs{O}^\lambda_x e^{\iunit \bs{H}^\lambda t'} \cdot \\
& \hspace{11em} \cdot \parens{\prod_{j=p}^1 e^{-\iunit \bs{H}^\lambda t_j} e^{-\iunit \theta_j\bs{A}^\lambda} e^{\iunit \bs{H}^\lambda t_j}} e^{-\iunit \bs{H}^\lambda t''} - \\
& \hspace{5em} - e^{\iunit \bs{H}^\lambda t''} \parens{\prod_{j=1}^p e^{-\iunit \bs{H}^\lambda t_j} e^{\iunit \theta_j \bs{A}^\lambda} e^{\iunit \bs{H}^\lambda t_j}}
e^{-\iunit \bs{H}^\lambda t'} \bs{O}^\lambda_x e^{\iunit \bs{H}^\lambda t'} \cdot \\
& \hspace{11em} \cdot \parens{\prod_{j=p}^1 e^{-\iunit \bs{H}^\lambda t_j} e^{-\iunit \theta_j\bs{A}^\lambda} \parens{\bs{A}^\lambda}^{\delta_{ij}} e^{\iunit \bs{H}^\lambda t_j}} e^{-\iunit \bs{H}^\lambda t''} \Bigg) \ket{x}.
\end{aligned}
\end{equation}

At $\bs{\theta} = 0$, this simplifies to
\begin{equation}
\partial_i \ell(\bs{\theta}) \Big|_{\bs{\theta} = 0} =
\frac{\iunit}{M} \sum_{\lambda=1}^{\Lambda'} \sum_{x=1}^{M_\lambda}
\bra{x} e^{\iunit \bs{H}^\lambda t''} \brackets{e^{-\iunit \bs{H}^\lambda t_i} \bs{A}^\lambda e^{\iunit \bs{H}^\lambda t_i}, \; e^{-\iunit \bs{H}^\lambda t'} \bs{O}^\lambda_x e^{\iunit \bs{H}^\lambda t'}} e^{-\iunit \bs{H}^\lambda t''} \ket{x}.
\end{equation}
\end{proof}

\newpage

\begin{restatable}{lemma}{restateETHMomentMatchingGradient}\label{lem:eth_moment_matching_gradient}
Consider a set of nonnegative integers $\{k_{\lambda, x, i, 1}, k_{\lambda, x, i, 2} : \lambda \in [\Lambda'], x \in [M_\lambda], i \in [p]\}$. Let
\begin{equation}
k = \sum_{\lambda=1}^{\Lambda'} \sum_{x=1}^{M_\lambda} \sum_{i=1}^p \left(k_{\lambda,x,i,1} + k_{\lambda,x,i,2}\right).
\end{equation}
Let $\bs{H}$ be a Hamiltonian satisfying the full Eigenstate Thermalization Hypothesis (Hypothesis~\ref{hyp:full_eth}). Then, there exists some $T = \poly(L, k)$ such that
\begin{equation}
\begin{aligned}
& \EE_{\substack{t_1, \dots, t_p, \\ t', t'' \sim \cU(0, T)}} \prod_{\lambda=1}^{\Lambda'} \prod_{x=1}^{M_\lambda} \prod_{i=1}^p \Tr\parens{e^{-\iunit \bs{H}^\lambda t''} \op{x}{x} e^{\iunit \bs{H}^\lambda t''} e^{-\iunit \bs{H}^\lambda t_i} \bs{A}^\lambda e^{\iunit \bs{H}^\lambda t_i} e^{-\iunit \bs{H}^\lambda t'} \bs{O}_x^\lambda e^{\iunit \bs{H}^\lambda t'}}^{k_{\lambda,x,i,1}} \\
& \hspace{8em} \Tr\parens{e^{-\iunit \bs{H}^\lambda t''} \op{x}{x} e^{\iunit \bs{H}^\lambda t''} e^{-\iunit \bs{H}^\lambda t'} \bs{O}_x^\lambda e^{\iunit \bs{H}^\lambda t'} e^{-\iunit \bs{H}^\lambda t_i} \bs{A}^\lambda e^{\iunit \bs{H}^\lambda t_i}}^{k_{\lambda,x,i,2}} = \\
={}& O(1 / N_\mathrm{min}^\ast) + \EE_{\substack{\bs{g}_1, \dots, \bs{g}_p, \\ \bs{g}', \bs{g}'' \sim \Haar}} \prod_{\lambda=1}^{\Lambda'} \prod_{x=1}^{M_\lambda} \prod_{i=1}^p
\Tr\parens{\bs{g}''^\lambda \op{x}{x} \bs{g}''^{\lambda \dagger} \bs{g}_i^\lambda \bs{A}^\lambda \bs{g}_i^{\lambda \dagger} \bs{g}'^\lambda \bs{O}_x^\lambda \bs{g}'^{\lambda \dagger}}^{k_{\lambda,x,i,1}} \\
& \hspace{10em} \Tr\parens{\bs{g}''^\lambda \op{x}{x} \bs{g}''^{\lambda \dagger} \bs{g}'^\lambda \bs{O}_x^\lambda \bs{g}'^{\lambda \dagger} \bs{g}_i^\lambda \bs{A}^\lambda \bs{g}_i^{\lambda \dagger}}^{k_{\lambda,x,i,2}},
\end{aligned}
\end{equation}
where
\begin{equation}
N_\mathrm{min}^\ast = \min_\lambda N_\lambda^\ast.
\end{equation}
Here, the Haar-random distribution is over the unitaries with the block-diagonal structure of the Krylov subspaces: equivalently, each $\bs{g}_i^\lambda, \bs{g}'^\lambda, \bs{g}''^\lambda$ is chosen to be Haar-random in each of the sectors.
\end{restatable}

\begin{proof}
We can justify the above claim using generally accepted physical assumptions as follows. Arguing similarly to \citet[Appendix C]{Fava_2025}, if we assume Hypothesis~\ref{hyp:full_eth}, we can say that for a sufficiently large time $T$,
\begin{equation}
\begin{aligned}
& \EE_{\substack{t_1, \dots, t_p, \\ t', t'' \sim \cU(0, T)}} \prod_{\lambda=1}^{\Lambda'} \prod_{x=1}^{M_\lambda} \prod_{i=1}^p \Tr\parens{e^{-\iunit \bs{H}^\lambda t''} \op{x}{x} e^{\iunit \bs{H}^\lambda t''} e^{-\iunit \bs{H}^\lambda t_i} \bs{A}^\lambda e^{\iunit \bs{H}^\lambda t_i} e^{-\iunit \bs{H}^\lambda t'} \bs{O}_x^\lambda e^{\iunit \bs{H}^\lambda t'}}^{k_{\lambda,x,i,1}} \\
& \hspace{8em} \Tr\parens{e^{-\iunit \bs{H}^\lambda t''} \op{x}{x} e^{\iunit \bs{H}^\lambda t''} e^{-\iunit \bs{H}^\lambda t'} \bs{O}_x^\lambda e^{\iunit \bs{H}^\lambda t'} e^{-\iunit \bs{H}^\lambda t_i} \bs{A}^\lambda e^{\iunit \bs{H}^\lambda t_i}}^{k_{\lambda,x,i,2}} = \\
={}& \prod_{\lambda=1}^{\Lambda'} \Bigg( O(1 / N_\lambda^\ast) + \\
& + \prod_{x=1}^{M_\lambda} \prod_{i=1}^p \EE_{\substack{t_1, \dots, t_p, \\ t', t'' \sim \cU(0, T)}} \brackets{\Tr\parens{e^{-\iunit \bs{H}^\lambda t''} \op{x}{x} e^{\iunit \bs{H}^\lambda t''} e^{-\iunit \bs{H}^\lambda t_i} \bs{A}^\lambda e^{\iunit \bs{H}^\lambda t_i} e^{-\iunit \bs{H}^\lambda t'} \bs{O}_x^\lambda e^{\iunit \bs{H}^\lambda t'}}}^{k_{\lambda,x,i,1}}  \\
& \hspace{1em} \EE_{\substack{t_1, \dots, t_p, \\ t', t'' \sim \cU(0, T)}} \brackets{\Tr\parens{e^{-\iunit \bs{H}^\lambda t''} \op{x}{x} e^{\iunit \bs{H}^\lambda t''} e^{-\iunit \bs{H}^\lambda t'} \bs{O}_x^\lambda e^{\iunit \bs{H}^\lambda t'} e^{-\iunit \bs{H}^\lambda t_i} \bs{A}^\lambda e^{\iunit \bs{H}^\lambda t_i}}}^{k_{\lambda,x,i,2}} \Bigg) = \\
={}& O(1 / N_\mathrm{min}^\ast) + \prod_{\lambda=1}^{\Lambda'} \prod_{x=1}^{M_\lambda} \prod_{i=1}^p \\
& \hspace{3em} \EE_{\substack{t_1, \dots, t_p, \\ t', t'' \sim \cU(0, T)}} \brackets{\Tr\parens{e^{-\iunit \bs{H}^\lambda t''} \op{x}{x} e^{\iunit \bs{H}^\lambda t''} e^{-\iunit \bs{H}^\lambda t_i} \bs{A}^\lambda e^{\iunit \bs{H}^\lambda t_i} e^{-\iunit \bs{H}^\lambda t'} \bs{O}_x^\lambda e^{\iunit \bs{H}^\lambda t'}}}^{k_{\lambda,x,i,1}} \\
& \hspace{3em} \EE_{\substack{t_1, \dots, t_p, \\ t', t'' \sim \cU(0, T)}} \brackets{\Tr\parens{e^{-\iunit \bs{H}^\lambda t''} \op{x}{x} e^{\iunit \bs{H}^\lambda t''} e^{-\iunit \bs{H}^\lambda t'} \bs{O}_x^\lambda e^{\iunit \bs{H}^\lambda t'} e^{-\iunit \bs{H}^\lambda t_i} \bs{A}^\lambda e^{\iunit \bs{H}^\lambda t_i}}}^{k_{\lambda,x,i,2}}.
\end{aligned}
\end{equation}

Now, using the results of \citet{Fava_2025}, we can write
\begin{equation}
\begin{aligned}
& \EE_{\substack{t_1, \dots, t_p, \\ t', t'' \sim \cU(0, T)}} \Tr\parens{e^{-\iunit \bs{H}^\lambda t''} \op{x}{x} e^{\iunit \bs{H}^\lambda t''} e^{-\iunit \bs{H}^\lambda t_i} \bs{A}^\lambda e^{\iunit \bs{H}^\lambda t_i} e^{-\iunit \bs{H}^\lambda t'} \bs{O}_x^\lambda e^{\iunit \bs{H}^\lambda t'}} = \\
={} & \frac{1}{(N_\lambda^\ast)^2} \Tr(\bs{A}^\lambda) \Tr(\bs{O}_x^\lambda).
\end{aligned}
\end{equation}

Now, when we replace the time evolution with the Haar-random unitaries, by free independence in the asymptotic limit \citep{Voiculescu_1991} we will get the same result as above up to an $O(1 / N_\mathrm{min}^\ast)$ error.
\end{proof}

\restateGradientHaarReduction*

\begin{proof}
Consider the joint distribution $\mathfrak{p}$ of the terms in the expression for the $\partial_i \ell_x^\lambda$ for $\bs{\theta} = 0$, as in \Cref{lem:1st_deriv_at_0}, considering the commutator expression as two separate terms. The terms are of the form
\begin{equation}
\bra{x} e^{\iunit \bs{H}^\lambda t''} e^{-\iunit \bs{H}^\lambda t_i} \bs{A}^\lambda e^{\iunit \bs{H}^\lambda t_i} e^{-\iunit \bs{H}^\lambda t'} \bs{O}^\lambda_x e^{\iunit \bs{H}^\lambda t'} e^{-\iunit \bs{H}^\lambda t''} \ket{x}
\end{equation}
and
\begin{equation}
\bra{x} e^{\iunit \bs{H}^\lambda t''} e^{-\iunit \bs{H}^\lambda t'} \bs{O}^\lambda_x e^{\iunit \bs{H}^\lambda t'} e^{-\iunit \bs{H}^\lambda t_i} \bs{A}^\lambda e^{\iunit \bs{H}^\lambda t_i} e^{-\iunit \bs{H}^\lambda t''} \ket{x}.
\end{equation}
and there are $d = 2$ of these terms, so the distribution is in $\RR^2$. We can also consider the joint distribution of the terms contributing to both $\partial_i \ell_x^\lambda(\bs{\theta})$ and $\partial_j \ell_y^{\lambda'}(\bs{\theta})$, in which case the distribution will be in $\RR^4$. Similarly, let $\mathfrak{q}$ be the joint distribution of the same terms, but with the time-evolved Hamiltonian replaced with the Haar-random matrices $\bs{g}_i, \bs{g}_j, \bs{g}', \bs{g}''$. Now, from \Cref{lem:eth_moment_matching_gradient}, we know that we can pick $T = \poly(L, k)$ such that mixed moments of order $k$ of the distributions $\mathfrak{p}$ and $\mathfrak{q}$ differ by an additive error of $\epsilon = 1 / N_\mathrm{min}^\ast$. The random variables corresponding to the terms are bounded, and thus they are subgaussian \citep{Vershynin_2018} and thus their moments of order $k$ are upper bounded by $\parens{C \sqrt{k}}^k$ for some constant $C$ independent of $k$.
Now, applying \Cref{lem:lp_bound_from_moments} gives us that the distribution $\mathfrak{p}$ converges to $\mathfrak{q}$ up to an error of
\begin{equation}
O\parens{\frac{\log \log N_\mathrm{min}}{\log N_\mathrm{min}} + \frac{\log k}{\sqrt{k}}} = o(1)
\end{equation}
in the L\'evy-Prokhorov metric (note that the $\mu$ term is subleading). Now, consider taking a general value of $\bs{\theta}$. Then, each of the terms becomes conjugated by some unitary matrix (see \Cref{lem:1st_deriv_at_0}). By the unitary invariance property of the Haar ensemble, the mixed moments of the distribution $\mathfrak{q}$ remain the same, which means that it approximates the joint distribution of the derivative terms at any value of $\bs{\theta}$. since the L\'evy-Prokhorov distance is preserved under the linear transformation of summing the terms to form the commutator expressions and is only reduced when dividing by the factor of $M$, the proof is complete.
\end{proof}

\restateGradientGaussianReduction*

\begin{proof}
We first perform the reduction to Haar-random matrices as in \Cref{lem:gradient_haar_reduction}. Now, we can proceed similarly to Lemma 27 of \citet{Anschuetz_2025}, based on the results of \citet{Jiang_2010}. From this, we incur an additional error of
\begin{equation}
O\parens{\sqrt{\frac{\log N^\ast}{N^\ast}}} = o(1)
\end{equation}
in the L\'evy-Prokhorov metric.
\end{proof}

\restateSemiIsotropicAReduction*

\begin{proof}
Consider each of the terms contributing to the gradient as previously. Replace each $\bs{A}^\lambda$ term with $\bs{A}^\lambda - \widetilde{\bs{A}}^\lambda$, so we have terms of the form
\begin{equation}
\begin{aligned}
& L^\lambda_{x,i} \coloneq \\
={} & \frac{1}{{M N_\lambda^\ast}^2} \parens{
\bra{x} \widetilde{\bs{g}}_i^\lambda \parens{\bs{A}^\lambda - \widetilde{\bs{A}}^\lambda} \widetilde{\bs{g}}_i^{\lambda \dagger} \widetilde{\bs{g}}'^\lambda \bs{O}^\lambda_x \widetilde{\bs{g}}'^{\lambda \dagger} \ket{x}
- \bra{x} \widetilde{\bs{g}}'^\lambda \bs{O}^\lambda_x \widetilde{\bs{g}}'^{\lambda \dagger} \widetilde{\bs{g}}_i^\lambda \parens{\bs{A}^\lambda - \widetilde{\bs{A}}^\lambda} \widetilde{\bs{g}}_i^{\lambda \dagger} \ket{x}
} = \\
={} & \frac{1}{M {N_\lambda^\ast}^2} \Tr\parens{\bs{M}^\lambda_x \widetilde{\bs{g}}_i^\lambda \parens{\bs{A}^\lambda - \widetilde{\bs{A}}^\lambda} \widetilde{\bs{g}}_i^{\lambda \dagger}},
\end{aligned}
\end{equation}
and
\begin{equation}
\bs{M}^\lambda_x = \widetilde{\bs{g}}'^\lambda \bs{O}^\lambda_x \widetilde{\bs{g}}'^{\lambda \dagger} \op{x}{x} - \op{x}{x} \widetilde{\bs{g}}'^\lambda \bs{O}^\lambda_x \widetilde{\bs{g}}'^{\lambda \dagger}
\end{equation}
is an anti-Hermitian traceless operator of rank $2$ and operator norm at most $2$. We can therefore write it in the form
\begin{equation}
\bs{M}^\lambda_x = \iunit \mu^\lambda_x \op{\psi^\lambda_{x}}{\psi^\lambda_{x}} - \iunit \mu^\lambda_x \op{\varphi^\lambda_{x}}{\varphi^\lambda_{x}}
\end{equation}
with $\abs{\mu_x^\lambda}$ bounded by $2$. Now, we can also write
\begin{equation}
\widetilde{\bs{g}}_i^\lambda \parens{\bs{A}^\lambda - \widetilde{\bs{A}}^\lambda} \widetilde{\bs{g}}_i^{\lambda \dagger} =
\sum_{k=1}^{N_\lambda^\ast} (a_k^\lambda - \widetilde{a}_k^\lambda) \op{\chi^\lambda_{i,k}}{\chi^\lambda_{i,k}},
\end{equation}
where the $a_k^\lambda, \widetilde{a}_k^\lambda$ are the eigenvalues of $\bs{A}^\lambda, \widetilde{\bs{A}}^\lambda$ in non-increasing order and $\ket{\chi^\lambda_{i,k}}$ is the $k$th column of $\widetilde{\bs{g}}_i^\lambda$ in the basis in which $\bs{A}^\lambda$ and $\widetilde{\bs{A}}^\lambda$ are both diagonal. So now, we have that
\begin{equation}
\begin{aligned}
& L^\lambda_{x,i} =
\frac{\iunit \mu^\lambda_x}{M {N_\lambda^\ast}^2} \sum_{k=1}^{N_\lambda^\ast} (a_k^\lambda - \widetilde{a}_k^\lambda) \parens{\abs{\braket{\chi^\lambda_{i,k}}{\psi^\lambda_{x}}}^2 - \abs{\braket{\chi^\lambda_{i,k}}{\varphi^\lambda_{x}}}^2}.
\end{aligned}
\end{equation}

Now, for any normalized vector $\ket{\psi}$, observe that $\abs{\braket{\chi^\lambda_{i,k}}{\psi}}^2$ is a sum of squared Gaussian random variables, so $\abs{\braket{\chi^\lambda_{i,k}}{\psi^\lambda_{x}}}^2 - \abs{\braket{\chi^\lambda_{i,k}}{\varphi^\lambda_{x}}}^2$ is subexponential. Then, by the version of Bernstein's inequality found in \citet{Vershynin_2018}, since $\abs{\mu^\lambda_x} \leq 2$, we have that
\begin{equation}
\begin{aligned}
& \Pr\brackets{\abs{L^\lambda_{x,i}} \geq \epsilon} \leq \\
& \Pr\brackets{
\abs{
\frac{\abs{\mu^\lambda_x}}{M {N_\lambda^\ast}^2} \sum_{k=1}^{N_\lambda^\ast} (a_k^\lambda - \widetilde{a}_k^\lambda) \parens{\abs{\braket{\chi^\lambda_{i,k}}{\psi^\lambda_{x}}}^2 - \abs{\braket{\chi^\lambda_{i,k}}{\varphi^\lambda_{x}}}^2}
} \geq \epsilon} \leq \\
\leq{} & 2 \exp \brackets{-c \min\parens{
\frac{\epsilon}{\frac{2}{M {N_\lambda^\ast}^2} \max_k \abs{a_k^\lambda - \widetilde{a}_k^\lambda}},
\frac{\epsilon^2}{\frac{4}{M^2 {N_\lambda^\ast}^4} \sum_k \parens{a_k^\lambda - \widetilde{a}_k^\lambda}^2}
}} \leq \\
\leq{} & 2 \exp \brackets{-c M {N_\lambda^\ast}^2 \min\parens{
\frac{\epsilon}{2 \max_k \abs{a_k^\lambda - \widetilde{a}_k^\lambda}},
\frac{\epsilon^2}{\frac{4}{M {N_\lambda^\ast}^2} \sum_k \parens{a_k^\lambda - \widetilde{a}_k^\lambda}^2}
}}
\end{aligned}
\end{equation}
for some constant $c > 0$. Now, observe that
\begin{equation}
\begin{aligned}
& \frac{1}{M {N_\lambda^\ast}^2} \sum_i \parens{a_k^\lambda - \widetilde{a}_k^\lambda}^2 =
\frac{1}{M {N_\lambda^\ast}^2} \Tr\brackets{\parens{\bs{A}^\lambda - \widetilde{\bs{A}}^\lambda}^2} = \\
={}& \frac{1}{M {N_\lambda^\ast}^2} \Tr\brackets{\parens{\bs{A}^\lambda}^2 + \parens{\widetilde{\bs{A}}^\lambda}^2 - \bs{A}^\lambda \widetilde{\bs{A}}^\lambda - \widetilde{\bs{A}}^\lambda \bs{A}^\lambda} =
\frac{2}{M {N_\lambda^\ast}^2} \parens{\Tr((\bs{A}^\lambda)^2) - \Tr\parens{\bs{A}^\lambda \widetilde{\bs{A}}^\lambda}}.
\end{aligned}
\end{equation}
Now, since
\begin{equation}
\Tr\parens{\bs{A}^\lambda \widetilde{\bs{A}}^\lambda} =
\frac{\Tr(\bs{A}^\lambda)}{r_A^\lambda} \sum_{i=1}^{r_A^\lambda} a_k^\lambda \geq
\frac{\Tr((\bs{A}^\lambda)^2)}{\Tr(\bs{A}^\lambda)} \parens{\Tr(\bs{A}^\lambda) - N_\lambda^\ast a_{r_A^\lambda + 1}^\lambda},
\end{equation}
we must have that
\begin{equation}
\frac{4}{M {N_\lambda^\ast}^2} \sum_i \parens{a_k^\lambda - \widetilde{a}_k^\lambda}^2 \leq
\frac{8}{M {N_\lambda^\ast}^2} \frac{\Tr((\bs{A}^\lambda)^2)}{\Tr(\bs{A}^\lambda)} N_\lambda^\ast a_{r_A^\lambda + 1}^\lambda = o(1),
\end{equation}
which holds since $\bs{A}$ is a local operator and its eigenvalues do not depend on the system size. We also have that
\begin{equation}
2 \max_k \abs{a_k^\lambda - \widetilde{a}_k^\lambda} =
O\parens{a_1^\lambda - \frac{\Tr((\bs{A}^\lambda)^2)}{\Tr(\bs{A}^\lambda)} a_{r_A^\lambda + 1}} = O(1).
\end{equation}
Now, we can just apply the union bound on all terms (of which there at most $2$ since we are considering the joint distribution of at most $2$ contributions) to say that
\begin{equation}
\Pr\brackets{\abs{L^\lambda_{x,i}} \geq \epsilon \vee \abs{L^{\lambda'}_{x',i'}} \geq \epsilon} \leq 4 \exp\brackets{-c M N_\mathrm{min}^{\ast 2} \min(\epsilon, \epsilon^2)}.
\end{equation}
Now, to solve for the Ky Fan metric, which upper-bounds the L\'evy-Prokhorov metric~\citep{Strassen_1965}, we need to find an upper bound on
\begin{equation}
\inf \braces{\epsilon > 0 : \Pr\brackets{\abs{L^\lambda_{x,i}} \geq \epsilon \vee \abs{L^{\lambda'}_{x',i'}} \geq \epsilon} \leq \epsilon}.
\end{equation}
For the case where $\epsilon < \epsilon^2$, equality occurs when
\begin{align}
4 \exp\parens{-c M N_\mathrm{min}^{\ast 2} \epsilon} &= \epsilon \\
\epsilon \exp\parens{c M N_\mathrm{min}^{\ast 2} \epsilon} &= 4 \\
c M N_\mathrm{min}^{\ast 2} \epsilon \exp\parens{c M N_\mathrm{min}^{\ast 2} \epsilon} &=
4 c M N_\mathrm{min}^{\ast 2} \\
c M N_\mathrm{min}^{\ast 2} \epsilon &=
W(4 c M N_\mathrm{min}^{\ast 2}),
\end{align}
where $W$ is the Lambert $W$ function. Since it is upper-bounded by the logarithm, an upper bound for the Ky Fan metric is
\begin{equation}
\begin{aligned}
c M N_\mathrm{min}^{\ast 2} \epsilon &\leq
\log(4 c M N_\mathrm{min}^{\ast 2}) \\
\epsilon &\leq
\frac{\log(4 c M N_\mathrm{min}^{\ast 2})}{c M N_\mathrm{min}^{\ast 2}} = O\parens{\frac{\log(M N_\mathrm{min})}{M N_\mathrm{min}^2}}.
\end{aligned}
\end{equation}
When $\epsilon^2 < \epsilon$ (when $\epsilon < 1$), equality occurs when
\begin{align}
4 \exp\parens{-c M N_\mathrm{min}^{\ast 2} \epsilon^2} &= \epsilon \\
c M N_\mathrm{min}^{\ast 2} \epsilon^2 &=
W(4 c M N_\mathrm{min}^{\ast 2} \epsilon),
\end{align}
so the upper bound for the Ky Fan metric (and thus the L\'evy-Prokhorov metric) is
\begin{align}
c M N_\mathrm{min}^{\ast 2} \epsilon^2 &\leq
\log(4 c M N_\mathrm{min}^{\ast 2} \epsilon) \leq
\log(4 c M N_\mathrm{min}^{\ast 2}) \\
\epsilon &= O\parens{\frac{\sqrt{\log(M N_\mathrm{min})}}{M N_\mathrm{min}^2}} = o(1).
\end{align}
Since $\epsilon < 1$ for large enough system sizes, this is the correct asymptotic bound.
\end{proof}

\restateGradientCalculation*

\begin{proof}
We know by the reductions above that the contribution to the derivative with respect to $\theta_i$ from entry $x$ in sector $\lambda$ can be approximated as
\begin{equation}
\hat{\ell}_{x; i}^\lambda = \frac{\iunit}{M {N_\lambda^\ast}^2} \bra{x} \widetilde{\bs{g}}''^{\lambda \dagger} \brackets{\widetilde{\bs{g}}_i^\lambda \widetilde{\bs{A}}^\lambda \widetilde{\bs{g}}_i^{\lambda \dagger}, \; \widetilde{\bs{g}}'^\lambda \bs{O}^\lambda_x \widetilde{\bs{g}}'^{\lambda \dagger}} \widetilde{\bs{g}}''^\lambda \ket{x}.
\end{equation}
By unitary invariance, we can replace $\widetilde{\bs{g}}_i^\lambda$ with $\widetilde{\bs{g}}'^\lambda \widetilde{\bs{g}}_i^\lambda$ while maintaining the same probability distribution (note that technically, we are applying this \emph{before} the reduction in \Cref{lem:gradient_gaussian_reduction} to the Haar-random matrices). But now, since
\begin{equation}
\begin{aligned}
& \brackets{\widetilde{\bs{g}}'^\lambda \widetilde{\bs{g}}_i^\lambda \widetilde{\bs{A}}^\lambda \widetilde{\bs{g}}_i^{\lambda \dagger} \widetilde{\bs{g}}'^{\lambda \dagger}, \; \widetilde{\bs{g}}'^\lambda \bs{O}^\lambda_x \widetilde{\bs{g}}'^{\lambda \dagger}} = \\
={} & \widetilde{\bs{g}}'^\lambda \widetilde{\bs{g}}_i^\lambda \widetilde{\bs{A}}^\lambda \widetilde{\bs{g}}_i^{\lambda \dagger} \bs{O}^\lambda_x \widetilde{\bs{g}}'^{\lambda \dagger} -
\widetilde{\bs{g}}'^\lambda \bs{O}^\lambda_x \widetilde{\bs{g}}_i^\lambda \widetilde{\bs{A}}^\lambda \widetilde{\bs{g}}_i^{\lambda \dagger} \widetilde{\bs{g}}'^{\lambda \dagger} = \\
={} & \widetilde{\bs{g}}'^\lambda
\brackets{\widetilde{\bs{g}}_i^\lambda \widetilde{\bs{A}} ^\lambda \widetilde{\bs{g}}_i^{\lambda \dagger}, \; \bs{O}^\lambda_x} \widetilde{\bs{g}}'^{\lambda \dagger},
\end{aligned}
\end{equation}
we can absorb $\bs{g}''^\lambda$ into $\bs{g}'^\lambda$ and write
\begin{equation}
\begin{aligned}
& \hat{\ell}_{x; i}^\lambda = \frac{\iunit}{M {N_\lambda^\ast}^2} \bra{x} \widetilde{\bs{g}}'^\lambda \brackets{\widetilde{\bs{g}}_i^\lambda \widetilde{\bs{A}}^\lambda \widetilde{\bs{g}}_i^{\lambda \dagger}, \; \bs{O}^\lambda_x} \widetilde{\bs{g}}'^{\lambda \dagger} \ket{x} = \\
={}& \frac{\iunit}{M {N_\lambda^\ast}^2}
\frac{\Tr((\bs{A}^\lambda)^2)}{\Tr(\bs{A}^\lambda)}
\sum_{\mu=1}^{r_A^\lambda} \bra{x} \widetilde{\bs{g}}'^\lambda \brackets{\widetilde{\bs{g}}_i^\lambda \op{\mu}{\mu} \widetilde{\bs{g}}_i^{\lambda \dagger}, \; \bs{O}^\lambda_x} \widetilde{\bs{g}}'^{\lambda \dagger} \ket{x}
\end{aligned}
\end{equation}

Now, we can take a matrix $\bs{X}^\lambda$ of dimension $N_\lambda^\ast \times (M_\lambda + p r_A^\lambda)$ with i.i.d. Gaussian entries, such that
\begin{equation}
\bs{X}^\lambda \ket{x} = \widetilde{\bs{g}}'^{\lambda \dagger} \ket{x}
\end{equation}
and
\begin{equation}
\bs{X}^\lambda \ket{i, \mu} = \widetilde{\bs{g}}_i^\lambda \ket{\mu}.
\end{equation}
Note that we are using $\{x\}$ and $\{(i, \mu)\}$ to label the $M_\lambda + p r_A^\lambda$ columns of $\bs{X}^\lambda$.

Now, we can write
\begin{equation}
\begin{aligned}
& \hat{\ell}_{x; i}^\lambda = \frac{\iunit \Tr((\bs{A}^\lambda)^2)}{M {N_\lambda^\ast}^2 \Tr(\bs{A}^\lambda)}
\sum_{\mu=1}^{r_A^\lambda}
\bra{x} \bs{X}^{\lambda \dagger} \brackets{\bs{X}^\lambda \op{i, \mu}{i, \mu} \bs{X}^{\lambda \dagger}, \bs{O}_x^\lambda} \bs{X}^\lambda \ket{x} = \\
={}& \frac{\iunit \Tr((\bs{A}^\lambda)^2)}{M {N_\lambda^\ast}^2 \Tr(\bs{A}^\lambda)}
\sum_{\mu=1}^{r_A^\lambda}
\Big(\bra{x}^\lambda \bs{X}^{\lambda \dagger} \bs{X}^\lambda \ket{i, \mu} \bra{i, \mu} \bs{X}^{\lambda \dagger} \bs{O}_x^\lambda \bs{X}^\lambda \ket{x} - \\
& \hspace{8em} - \bra{x} \bs{X}^{\lambda \dagger} \bs{O}_x^\lambda \bs{X}^\lambda \ket{i, \mu} \bra{i, \mu} \bs{X}^{\lambda \dagger} \bs{X}^\lambda \ket{x}\Big) = \\
={}& -\frac{\iunit \Tr((\bs{A}^\lambda)^2)}{M {N_\lambda^\ast}^2 \Tr(\bs{A}^\lambda)}
\sum_{\mu=1}^{r_A^\lambda}
\parens{\bra{x} \bs{W}^\lambda \ket{i, \mu} \bra{i, \mu} \bs{W}^{\lambda,x} \ket{x} -
\bra{x} \bs{W}^{\lambda,x} \ket{i, \mu} \bra{i, \mu} \bs{W}^\lambda \ket{x}} = \\
={}& -\frac{\iunit \Tr((\bs{A}^\lambda)^2)}{M {N_\lambda^\ast}^2 \Tr(\bs{A}^\lambda)} \sum_{\mu=1}^{r_A^\lambda} \sum_{y=1}^{M_\lambda + 1} \Big(\bra{x} \bs{W}^{\lambda,y} \ket{i, \mu} \bra{i, \mu} \bs{W}^{\lambda,x} \ket{x} - \\
& \hspace{8em} - \bra{x} \bs{W}^{\lambda,x} \ket{i,\mu} \bra{i,\mu} \bs{W}^{\lambda,y} \ket{x}\Big) = \\
={}& -\frac{\iunit \Tr((\bs{A}^\lambda)^2)}{M {N_\lambda^\ast}^2 \Tr(\bs{A}^\lambda)} \sum_{\mu=1}^{r_A^\lambda} \sum_{y=1}^{M_\lambda + 1} (-2 i) \Im\parens{\bra{x} \bs{W}^{\lambda,x} \ket{i, \mu} \bra{i, \mu} \bs{W}^{\lambda,y} \ket{x}} = \\
={}& -\frac{2 \Tr((\bs{A}^\lambda)^2)}{M {N_\lambda^\ast}^2 \Tr(\bs{A}^\lambda)} \sum_{\mu=1}^{r_A^\lambda} \sum_{y=1}^{M_\lambda + 1} \Im\parens{\bra{x} \bs{W}^{\lambda,x} \ket{i, \mu} \bra{i, \mu} \bs{W}^{\lambda,y} \ket{x}}.
\end{aligned}
\end{equation}

Here, we define $\bs{W}^{\lambda,x} \coloneq -\bs{X}^{\lambda \dagger} \bs{O}_x^\lambda \bs{X}^\lambda$ for $1 \leq x \leq M_\lambda$ (recall that $\bs{O}_x^\lambda$ is negative semidefinite), and we define
\begin{equation}
\bs{W}^{\lambda, M_\lambda + 1} \coloneq \bs{W}^\lambda - \sum_{x=1}^{M_\lambda} \bs{W}^{\lambda,x}
\end{equation}
and $\bs{W}^\lambda \coloneq \bs{X}^{\lambda \dagger} \bs{X}^{\lambda}$.

Now, we can write this in terms of the entries of $\bs{X}^{\lambda,x}$ which we take to be an $r_x \times (M_\lambda + p r_A^\lambda)$ submatrix of $\bs{X}$. Here for $1 \leq x \leq M_\lambda$ we have $r_x = \rank(\bs{O}_x^\lambda) = N_\lambda$, and $r_{M_\lambda + 1} = N_\lambda^\ast - M_\lambda N_\lambda$.

\begin{equation}
\hat{\ell}_{x; i}^\lambda = -\frac{2 \Tr((\bs{A}^\lambda)^2)}{M {N_\lambda^\ast}^2 \Tr(\bs{A}^\lambda)} \sum_{\mu=1}^{r_A^\lambda} \sum_{y \neq x} \Im\parens{\sum_{j=1}^{r_x} \parens{X^{\lambda,x}_{j,x}}^\ast X^{\lambda,x}_{j,(i,\mu)} \sum_{k=1}^{r_y} \parens{X^{\lambda,y}_{k,(i,\mu)}}^\ast X^{\lambda,y}_{k,x}}.
\end{equation}

Now, observe that for $\lambda \neq \lambda'$ or $x \neq x'$, we must have by a symmetry argument that
\begin{equation}
\EE\brackets{(\hat{\ell}_{x; i}^\lambda) (\hat{\ell}_{x'; i}^{\lambda'})} = 0,
\end{equation}
since each term in the expansion of this will have a factor that is a Gaussian random variable independent from the rest. Also recall that this is a valid approximation for the corresponding term in the actual loss since we ensured in our reductions that the joint distribution of two contributions such as this converges in distribution to our approximation.

Now,
\begin{equation}
\begin{aligned}
& \EE\brackets{(\hat{\ell}_{x; i})^2} = \\
={} & \parens{\frac{2 \Tr((\bs{A}^\lambda)^2)}{M {N_\lambda^\ast}^2 \Tr(\bs{A}^\lambda)}}^2 \sum_{\mu=1}^{r_A^\lambda} \sum_{\substack{y = 1 \\ y \neq x}}^{M_\lambda + 1}
\EE\brackets{\Im\parens{\sum_{j=1}^{r_x} \parens{X^{\lambda,x}_{j,x}}^\ast X^{\lambda,x}_{j,(i,\mu)} \sum_{k=1}^{r_y} \parens{X^{\lambda,y}_{k,(i,\mu)}}^\ast X^{\lambda,y}_{k,x}}^2} = \\
={}& \frac{1}{2} \parens{\frac{2 \Tr((\bs{A}^\lambda)^2)}{M {N_\lambda^\ast}^2 \Tr(\bs{A}^\lambda)}}^2 \sum_{\mu=1}^{r_A^\lambda} \sum_{\substack{y = 1 \\ y \neq x}}^{M_\lambda + 1}
\EE\brackets{\abs{\sum_{j=1}^{r_x} \parens{X^{\lambda,x}_{j,x}}^\ast X^{\lambda,x}_{j,(i,\mu)} \sum_{k=1}^{r_y} \parens{X^{\lambda,y}_{k,(i,\mu)}}^\ast X^{\lambda,y}_{k,x}}^2}.
\end{aligned}
\end{equation}
Note that all cross terms in the above calculation get eliminated due to symmetry in the Gaussian distributions of the unmatched matrix elements.

Now, observe that
\begin{equation}
\begin{aligned}
& \EE\brackets{\abs{\sum_{j=1}^{r_x} \parens{X^{\lambda,x}_{j,x}}^\ast X^{\lambda,x}_{j,(i,\mu)} \sum_{k=1}^{r_y} \parens{X^{\lambda,y}_{k,(i,\mu)}}^\ast X^{\lambda,y}_{k,x}}^2} = \\
={}& \EE\brackets{\abs{\sum_{j=1}^{r_x} \parens{X^{\lambda,x}_{j,x}}^\ast X^{\lambda,x}_{j,(i,\mu)}}^2}
\EE\brackets{\abs{\sum_{k=1}^{r_y} \parens{X^{\lambda,y}_{k,(i,\mu)}}^\ast X^{\lambda,y}_{k,x}}^2} = \\
={}& \EE\brackets{\sum_{j=1}^{r_x} \abs{X^{x,\lambda}_{j,x}}^2 \abs{X^{\lambda,x}_{j,(i,\mu)}}^2}
\EE\brackets{\sum_{k=1}^{r_y} \abs{X^{\lambda,y}_{k,(i,\mu)}}^2 \abs{X^{\lambda,y}_{k,x}}^2} =
r_x r_y = N_\lambda r_y.
\end{aligned}
\end{equation}

Then, since
\begin{equation}
\sum_{\substack{y = 1 \\ y \neq x}}^{N_\lambda + 1} r_y = N_\lambda^\ast - N_\lambda = N_\lambda (N_a - 1),
\end{equation}
we have that
\begin{equation}
\begin{aligned}
& \EE\brackets{(\hat{\ell}_{x; i})^2} =
\frac{1}{2} \parens{\frac{2 \Tr((\bs{A}^\lambda)^2)}{M {N_\lambda^\ast}^2 \Tr(\bs{A}^\lambda)}}^2 r_A^\lambda (N_a - 1) N_\lambda^2 = \\
={}& \frac{2 \Tr((\bs{A}^\lambda)^2)^2}{M^2 {N_\lambda^\ast}^4 \Tr(\bs{A}^\lambda)^2} \frac{\Tr(\bs{A}^\lambda)^2}{\Tr((\bs{A}^\lambda)^2)} (N_a - 1) N_\lambda^2 =
\frac{2 (N_a - 1) \Tr((\bs{A}^\lambda)^2)}{M^2 N_a^4 N_\lambda^2}
\end{aligned}
\end{equation}

Thus,
\begin{equation}
\begin{aligned}
& \EE\brackets{\norm{\widehat{\nabla \ell}}^2} =
\sum_{i=1}^p \EE\brackets{(\hat{\ell}_{i})^2} =
\sum_{i=1}^p \sum_{\lambda=1}^{\Lambda'} \sum_{x=1}^{M_\lambda} \EE\brackets{(\hat{\ell}_{x; i})^2} = \\
={} & \sum_{\lambda=1}^{\Lambda'} \frac{2 M_\lambda (N_a - 1) \Tr((\bs{A}^\lambda)^2) p}{M^2 N_a^4 N_\lambda^2} = \Omega\parens{\frac{1}{\poly(L)}}.
\end{aligned}
\end{equation}
\end{proof}

\restateHessianCalculation*

\begin{proof}
We can write the second derivative of the loss function as:

\begin{equation}
\begin{aligned}
& \partial_i \partial_j \ell(\bs{\theta}) =
\frac{1}{M} \sum_{\lambda=1}^{\Lambda'} \sum_{x=1}^{M_\lambda}
\frac{\partial}{\partial \theta_i} \frac{\partial}{\partial \theta_j} \bra{x} \bs{U}^\lambda(\bs{\theta}) \bs{O}^\lambda_x \bs{U}^\lambda(\bs{\theta})^\dagger \ket{x} = \\
={}& -\frac{1}{M} \sum_{\lambda=1}^{\Lambda'} \sum_{x=1}^{M_\lambda}
\bra{x} e^{\iunit \bs{H}^\lambda t''} \Bigg( \parens{\prod_{k=1}^p e^{-\iunit \bs{H}^\lambda t_k} e^{\theta_k \bs{A}^\lambda} \parens{\bs{A}^\lambda}^{\delta_{ik} + \delta_{jk}} e^{\iunit \bs{H}^\lambda t_k}}
e^{-\iunit \bs{H}^\lambda t'} \bs{O}^\lambda_x e^{\iunit \bs{H}^\lambda t'} \cdot \\
& \hspace{11em} \cdot \parens{\prod_{k=p}^1 e^{-\iunit \bs{H}^\lambda t_k} e^{-\iunit \theta_k\bs{A}^\lambda} e^{\iunit \bs{H}^\lambda t_k}} - \\
& \qquad \qquad - \parens{\prod_{k=1}^p e^{-\iunit \bs{H}^\lambda t_k} e^{\theta_k \bs{A}^\lambda} \parens{\bs{A}^\lambda}^{\delta_{ik}} e^{\iunit \bs{H}^\lambda t_k}}
e^{-\iunit \bs{H}^\lambda t'} \bs{O}^\lambda_x e^{\iunit \bs{H}^\lambda t'} \cdot \\
& \hspace{11em} \cdot \parens{\prod_{k=p}^1 e^{-\iunit \bs{H}^\lambda t_k} e^{-\iunit \theta_k\bs{A}^\lambda} \parens{\bs{A}^\lambda}^{\delta_{jk}} e^{\iunit \bs{H}^\lambda t_k}} - \\
& \qquad \qquad - \parens{\prod_{k=1}^p e^{-\iunit \bs{H}^\lambda t_k} e^{\theta_k \bs{A}^\lambda} \parens{\bs{A}^\lambda}^{\delta_{jk}} e^{\iunit \bs{H}^\lambda t_k}}
e^{-\iunit \bs{H}^\lambda t'} \bs{O}^\lambda_x e^{\iunit \bs{H}^\lambda t'} \cdot \\
& \hspace{11em} \cdot \parens{\prod_{k=p}^1 e^{-\iunit \bs{H}^\lambda t_k} e^{-\iunit \theta_k\bs{A}^\lambda} \parens{\bs{A}^\lambda}^{\delta_{ik}} e^{\iunit \bs{H}^\lambda t_k}} + \\
& \qquad \qquad + \parens{\prod_{k=1}^p e^{-\iunit \bs{H}^\lambda t_k} e^{\theta_k \bs{A}^\lambda} e^{\iunit \bs{H}^\lambda t_k}}
e^{-\iunit \bs{H}^\lambda t'} \bs{O}^\lambda_x e^{\iunit \bs{H}^\lambda t'} \cdot \\
& \hspace{11em} \cdot \parens{\prod_{k=p}^1 e^{-\iunit \bs{H}^\lambda t_k} e^{-\iunit \theta_k\bs{A}^\lambda} \parens{\bs{A}^\lambda}^{\delta_{ik} + \delta_{jk}} e^{\iunit \bs{H}^\lambda t_k}} \Bigg) e^{-\iunit \bs{H}^\lambda t''} \ket{x}.
\end{aligned}
\end{equation}

Now, for every $i$, let
\begin{equation}
\bs{V}_i \coloneq \prod_{k=1}^i e^{-\iunit \bs{H}^\lambda t_k} e^{\theta_k \bs{A}^\lambda} e^{\iunit \bs{H}^\lambda t_k}.
\end{equation}
Then, assuming $i \geq j$ (without loss of generality, since $\partial_i \partial_j \ell(\bs{\theta}) = \partial_j \partial_i \ell(\bs{\theta})$), we can rewrite the above as follows:
\begin{equation}
\begin{aligned}
& \partial_i \partial_j \ell(\bs{\theta}) =
-\frac{1}{M} \sum_{\lambda=1}^{\Lambda'} \sum_{x=1}^{M_\lambda}
\bra{x} e^{\iunit \bs{H}^\lambda t''} \cdot \\
& \qquad \qquad \cdot \Big(
\bs{V}_j e^{-\iunit \bs{H}^\lambda t_j} \bs{A}^\lambda e^{\iunit \bs{H}^\lambda t_j} \bs{V}_j^\dagger
\bs{V}_i e^{-\iunit \bs{H}^\lambda t_i} \bs{A}^\lambda e^{\iunit \bs{H}^\lambda t_i} \bs{V}_i^\dagger \bs{V}_p
e^{-\iunit \bs{H}^\lambda t'} \bs{O}^\lambda_x e^{\iunit \bs{H}^\lambda t'}
\bs{V}_p^\dagger - \\
& \qquad \qquad - \bs{V}_i e^{-\iunit \bs{H}^\lambda t_i} \bs{A}^\lambda e^{\iunit \bs{H}^\lambda t_i} \bs{V}_i^\dagger \bs{V}_p
e^{-\iunit \bs{H}^\lambda t'} \bs{O}^\lambda_x e^{\iunit \bs{H}^\lambda t'}
\bs{V}_p^\dagger \bs{V}_j e^{-\iunit \bs{H}^\lambda t_j} \bs{A}^\lambda e^{\iunit \bs{H}^\lambda t_j} \bs{V}_j^\dagger - \\
& \qquad \qquad - \bs{V}_j e^{-\iunit \bs{H}^\lambda t_j} \bs{A}^\lambda e^{\iunit \bs{H}^\lambda t_j} \bs{V}_j^\dagger \bs{V}_p
e^{-\iunit \bs{H}^\lambda t'} \bs{O}^\lambda_x e^{\iunit \bs{H}^\lambda t'}
\bs{V}_p^\dagger \bs{V}_i e^{-\iunit \bs{H}^\lambda t_i} \bs{A}^\lambda e^{\iunit \bs{H}^\lambda t_i} \bs{V}_i^\dagger + \\
& \qquad \qquad + \bs{V}_p
e^{-\iunit \bs{H}^\lambda t'} \bs{O}^\lambda_x e^{\iunit \bs{H}^\lambda t'}
\bs{V}_p^\dagger \bs{V}_i e^{-\iunit \bs{H}^\lambda t_i} \bs{A}^\lambda e^{\iunit \bs{H}^\lambda t_i} \bs{V}_i^\dagger \bs{V}_j e^{-\iunit \bs{H}^\lambda t_j} \bs{A}^\lambda e^{\iunit \bs{H}^\lambda t_j} \bs{V}_j^\dagger \Big) \cdot \\
& \qquad \qquad \cdot e^{-\iunit \bs{H}^\lambda t''} \ket{x} = \\
={}& -\frac{1}{M} \sum_{\lambda=1}^{\Lambda'} \sum_{x=1}^{M_\lambda}
\bra{x} e^{\iunit \bs{H}^\lambda t''} \cdot \\
& \quad \cdot \brackets{
\bs{V}_j e^{-\iunit \bs{H}^\lambda t_j} \bs{A}^\lambda e^{\iunit \bs{H}^\lambda t_j} \bs{V}_j^\dagger, \;
\brackets{
\bs{V}_i e^{-\iunit \bs{H}^\lambda t_i} \bs{A}^\lambda e^{\iunit \bs{H}^\lambda t_i} \bs{V}_i^\dagger, \;
\bs{V}_p e^{-\iunit \bs{H}^\lambda t'} \bs{O}^\lambda_x e^{\iunit \bs{H}^\lambda t'} \bs{V}_p^\dagger
}} \cdot \\
& \quad \cdot e^{-\iunit \bs{H}^\lambda t''} \ket{x}.
\end{aligned}
\end{equation}

Now, analogously to \Cref{thm:gradient_calculation}, we let $\bs{X}^\lambda$ be a matrix of dimension $N_\lambda^\ast \times (M_\lambda + p N_\lambda^\ast)$ (note that we are no longer doing the semi-isotropic reduction on $\bs{A}$) such that:
\begin{equation}
\bs{X}^\lambda \ket{x} = e^{-\iunit \bs{H}^\lambda t''} \ket{x}
\end{equation}
and
\begin{equation}
\bs{X}^\lambda \ket{i, \mu} = \bs{V}_i e^{-\iunit \bs{H}^\lambda t_i} \ket{\mu}.
\end{equation}
Let $\widetilde{\bs{O}}^\lambda_x = \bs{V}_p e^{-\iunit \bs{H}^\lambda t'} \bs{O}^\lambda_x e^{\iunit \bs{H}^\lambda t'} \bs{V}_p^\dagger$, $\bs{W}^\lambda = \bs{X}^{\lambda \dagger} \bs{X}^\lambda$, and $\bs{W}^{\lambda,x} = -\bs{X}^{\lambda \dagger} \widetilde{\bs{O}}^\lambda_x \bs{X}^\lambda$. Then, we can write:

\begin{equation}
\begin{aligned}
& \partial_i \partial_j \ell(\bs{\theta}) = \\
={} & -\frac{1}{M} \sum_{\lambda=1}^{\Lambda'} \sum_{x=1}^{M_\lambda} \sum_{\mu=1}^{N_\lambda^\ast} \sum_{\nu=1}^{N_\lambda^\ast} \\
& \hspace{8em} a^\lambda_\mu a^\lambda_\nu
\bra{x} \bs{X}^{\lambda \dagger} \brackets{\bs{X}^\lambda \ket{j,\nu}\bra{j,\nu}\bs{X}^{\lambda \dagger}, \brackets{\bs{X}^\lambda \ket{i,\mu}\bra{i,\mu} \bs{X}^{\lambda \dagger}, \widetilde{\bs{O}}^\lambda_x}} \bs{X}^\lambda \ket{x} = \\
={}& \frac{1}{M} \sum_{\lambda=1}^{\Lambda'} \sum_{x=1}^{M_\lambda} \sum_{\mu=1}^{N_\lambda^\ast} \sum_{\nu=1}^{N_\lambda^\ast}
a^\lambda_\mu a^\lambda_\nu
\left(\bra{x} \bs{W}^\lambda \ket{j, \nu} \bra{j, \nu} \bs{W}^\lambda \ket{i, \mu} \bra{i, \mu} \bs{W}^{\lambda,x} \ket{x} - \right. \\
& \hspace{10em} - \bra{x} \bs{W}^\lambda \ket{j, \nu} \bra{j, \nu} \bs{W}^{\lambda,x} \ket{i, \mu} \bra{i, \mu} \bs{W}^\lambda \ket{x} - \\
& \hspace{10em} - \bra{x} \bs{W}^\lambda \ket{i, \mu} \bra{i, \mu} \bs{W}^{\lambda,x} \ket{j, \nu} \bra{j, \nu} \bs{W}^\lambda \ket{x} + \\
& \hspace{10em} \left. + \bra{x} \bs{W}^{\lambda,x} \ket{i, \mu} \bra{i, \mu} \bs{W}^\lambda \ket{j, \nu} \bra{j, \nu} \bs{W}^\lambda \ket{x} \right) = \\
={}& -\frac{2}{M} \sum_{\lambda=1}^{\Lambda'} \sum_{x=1}^{M_\lambda} \sum_{\mu=1}^{N_\lambda^\ast} \sum_{\nu=1}^{N_\lambda^\ast}
a^\lambda_\mu a^\lambda_\nu
\Re \left(\bra{x} \bs{W}^\lambda \ket{i, \mu} \bra{i, \mu} \bs{W}^{\lambda,x} \ket{j, \nu} \bra{j, \nu} \bs{W}^\lambda \ket{x} - \right. \\
& \hspace{10em} \left. - \bra{x} \bs{W}^{\lambda,x} \ket{i, \mu} \bra{i, \mu} \bs{W}^\lambda \ket{j, \nu} \bra{j, \nu} \bs{W}^\lambda \ket{x} \right) =
\\
={}& -\frac{2}{M} \sum_{\lambda=1}^{\Lambda'} \sum_{x=1}^{M_\lambda} \sum_{\mu=1}^{N_\lambda^\ast} \sum_{\nu=1}^{N_\lambda^\ast}
a^\lambda_\mu a^\lambda_\nu \Re\parens{
\bs{S}^{\lambda,x,\mu,\nu,1}_{i,j} \bs{W}^{\lambda,x,\mu,\nu}_{i,j} - \bs{S}^{\lambda,x,\mu,\nu,2}_{i,j} \bs{W}^{\lambda,\mu,\nu}_{i,j}},
\end{aligned}
\end{equation}
where
\begin{align}
\bs{W}^{\lambda,x,\mu,\nu}_{i,j} &=
\bra{i,\mu} \bs{W}^{\lambda,x} \ket{j,\nu} \\
\bs{W}^{\lambda,\mu,\nu}_{i,j} &=
\bra{i,\mu} \bs{W}^\lambda \ket{j,\nu} \\
\bs{S}^{\lambda,x,\mu,\nu,1}_{i,j} &=
\bra{x} \bs{W}^\lambda \ket{i, \mu} \bra{j, \nu} \bs{W}^\lambda \ket{x}
\\
\bs{S}^{\lambda,x,\mu,\nu,2}_{i,j} &=
\bra{x} \bs{W}^{\lambda,x} \ket{i, \mu} \bra{j, \nu} \bs{W}^\lambda \ket{x}.
\end{align}

So, we can write the Hessian as
\begin{equation}
\nabla^2 \ell = -\frac{2}{M} \sum_{\lambda=1}^{\Lambda'} \sum_{x=1}^{M_\lambda} \sum_{\mu=1}^{N_\lambda^\ast} \sum_{\nu=1}^{N_\lambda^\ast}
a^\lambda_\mu a^\lambda_\nu
\Re\parens{\bs{S}^{\lambda,x,\mu,\nu,1} \odot \bs{W}^{\lambda,x,\mu,\nu} - \bs{S}^{\lambda,x,\mu,\nu,2} \odot \bs{W}^{\lambda,\mu,\nu}}.
\end{equation}

Now, we know that
\begin{equation}
\rank(\bs{W}^{\lambda,x,\mu,\nu}) \leq
\rank(\bs{W}^{\lambda,\mu,\nu}) \leq
\rank(\bs{W}^\lambda) \leq N_\lambda^\ast.
\end{equation}
Additionally, it is clear that $\bs{S}^{\lambda,x,\mu,\nu,1}$ and $\bs{S}^{\lambda,x,\mu,\nu,2}$ both have rank $1$: their entries are products of a term only depending on $i$ and a term only depending on $j$, so they can be written as outer products of two vectors. Now, it is known that the rank of a Hadamard product of two matrices is upper-bounded by the product of their ranks, which means that
\begin{equation}
\rank\parens{\nabla^2 \ell} \leq
\sum_{\lambda=1}^{\Lambda'} \sum_{x=1}^{M_\lambda} \sum_{\mu=1}^{N_\lambda^\ast} \sum_{\nu=1}^{N_\lambda^\ast} 2 N_\lambda^\ast =
\sum_{\lambda=1}^{\Lambda'} 2 M_\lambda (N_\lambda^\ast)^3 \leq
2 M {N^\ast}^3 = 2 M N_a^3 N^3.
\end{equation}
Thus, if we choose $p \geq \rank\parens{\nabla^2 \ell}$, it must be the case that $\nabla^2 \ell$, which is a $p \times p$ matrix, does not have full rank.
\end{proof}

\begin{theorem}[No spurious local minima, formal version of \Cref{thm:local_minima}]\label{thm:formal_loc_min}
    For some $p=O(\poly(L))$, define:
    \begin{equation}
        \tilde{\ell}(\bs{\theta}) \coloneq \ell(\bs{\theta}) + \xi(\bs{\theta})
    \end{equation}
    where we have defined the perturbation:
    \begin{equation}
    \xi \parens{\bs{\theta}} \coloneq \epsilon \sum_{i=1}^p \parens{a_i \cos(\theta_i) + b_i \sin(\theta_i)},
    \end{equation}
    where $\epsilon > 0$, $\bs{\theta}\in[0, 2\pi)^p$, and $a_i, b_i$ are i.i.d. standard Gaussian (independent from the coefficients in the QNN ansatz). the expected number of critical points $\tilde{\ell}$ above the global minimum approaches $0$ as $\epsilon\to 0^+$.
\end{theorem}

\begin{proof}
We choose $p$ sufficiently large as in \Cref{prop:hessian_calculation} so that the Hessian becomes singular. Now, consider the perturbation term $\xi : [0, 2\pi)^p \rarr \RR$ as defined in the theorem statement. Observe that $\xi$ is a smooth function on the $p$-torus and its distribution is rotationally invariant. Furthermore, the second derivative of $\xi$ with respect to any $\theta_i$ is nonzero with high probability, and the marginal probability density of the gradient of $\xi$ at any $\bs{\theta}$ is continuous. If we consider $\tilde{\ell}(\bs{\theta}) = \ell(\bs{\theta}) + \xi(\bs{\theta})$, we can now check the following conditions for applying the Kac-Rice formula~\citep{Adler_2007}, as stated in \citet{Anschuetz_2022_Critical_Points}.
% Furthermore, observe that
% \begin{equation}
% \begin{aligned}
% & a_i \cos(\theta_i + \phi) + b_i \sin(\theta_i + \phi) = \\
% ={} & a_i \parens{\cos(\theta_i)\cos(\phi) - \sin(\theta_i)\sin(\phi)} + b_i \parens{\sin(\theta_i)\cos(\phi) + \cos(\theta_i)\sin(\phi)} = \\
% ={} & \parens{a_i \cos(\phi) + b_i \sin(\phi)} \cos(\theta_i) + \parens{-a_i \sin(\phi) + b_i \cos(\phi)} \sin(\theta_i),
% \end{aligned}
% \end{equation}
% so shifting $\theta_i$ transforms $a_i$ and $b_i$ by an orthogonal transformation. Since they are standard Gaussian, the distribution is unchanged, and thus $\xi(\bs{\theta})$ has the same distribution of functions as $\xi(\bs{\theta} + \bs{\phi})$ for any $\bs{\phi} \in \RR^p$.

\begin{enumerate}
    \item All components of $\tilde{\ell}, \nabla \tilde{\ell}, \nabla^2 \tilde{\ell}$ are almost surely continuous and have finite variances over $[0, 2\pi)^p$.
    
    \item At any $\bs{\theta} \in [0, 2\pi)^p$, the marginal probability density $\mathfrak{p} \parens{\nabla \tilde{\ell} (\bs{\theta})}$ is continuous at \\ $\nabla \tilde{\ell} (\bs{\theta)} = 0$.
    
    \item The conditional density of $\nabla \tilde{\ell} (\bs{\theta)}$ conditioned on $\ell(\bs{\theta}), \nabla^2 \tilde{\ell}(\bs{\theta})$ is bounded above and continuous at $\nabla \tilde{\ell} (\bs{\theta)} = 0$, uniformly in $\bs{\theta}$.

    \item The conditional density of $\det\parens{\nabla^2 \tilde{\ell} (\bs{\theta})}$ conditioned on $\nabla \tilde{\ell} (\bs{\theta)} = 0$ is continuous at $\det\parens{\nabla^2 \tilde{\ell} (\bs{\theta})} = 0$, uniformly in $\bs{\theta}$.
    
    \item The conditional density of $\tilde{\ell}(\bs{\theta})$ conditioned on $\nabla \tilde{\ell} (\bs{\theta)} = 0$ is continuous for values of $\tilde{\ell} (\bs{\theta})$ and $\nabla \tilde{\ell} (\bs{\theta)}$ in a neighborhood of $0$, uniformly in $\bs{\theta}$.

    \item The $p$th moments of the entries of the Hessian are bounded, uniformly in $\bs{\theta}$.

    \item The moduli of continuity $\omega$ of the components of $\tilde{\ell}(\bs{\theta}), \nabla \tilde{\ell} (\bs{\theta)}, \nabla^2 \tilde{\ell} (\bs{\theta)}$ satisfy $\Pr\brackets{\omega(\eta) > \delta} = o(\eta^p)$ for all $\delta > 0$ as $\eta \rarr 0^+$.
\end{enumerate}

\citet{Anschuetz_2022_Critical_Points} demonstrated that these assumptions hold for variational loss landscapes; for us, the main difference is ensuring conditions (4) and (5), since we are looking at the overparameterized regime where the original Hessian is singular.

(1) follows by the construction of the QNN ansatz, (2) and (3) can be seen by considering the approximation to the distribution of the gradient in \Cref{thm:gradient_calculation}. (4) is ensured by adding the perturbation term $\xi(\bs{\theta})$ to the loss, which makes the Hessian non-singular with high probability. Similarly, for condition (5), the perturbation term breaks any symmetry in critical manifolds with high probability, thus making the conditional density of $\tilde{\ell}(\bs{\theta})$ continuous. (6) holds since the second derivatives are uniformly bounded. (7) holds by the argument in \citet{Anschuetz_2022_Critical_Points}.

We can now apply the Kac-Rice formula, which says that the expected number of local minima with loss values greater than the optimum $\ell_\ast$ is given by
\begin{equation}
\int_{[0, 2\pi)^p} \mathfrak{p} \parens{\nabla \tilde{\ell} (\bs{\theta}) = 0} \EE\brackets{\abs{\det\parens{\nabla^2 \tilde{\ell} (\bs{\theta})}} \mathbbm{1}\brackets{\tilde{\ell}(\bs{\theta}) > \ell_\ast \wedge \nabla^2 \tilde{\ell} (\bs{\theta}) \succeq 0} \mid \nabla \tilde{\ell} (\bs{\theta}) = 0} \diff \bs{\theta}.
\end{equation}
Since $\epsilon$ can be made arbitrarily small and the determinant is a continuous function of the entries of the Hessian, the expectation of the determinant in the integrand above can also be made arbitrarily small. Since the density of the gradient at zero is bounded, and the gradient is not affected by the perturbation to leading order in $\epsilon$, the probability density term in the integrand remains bounded. Thus, the statement follows.
\end{proof}

\section{Details on the L\'evy-Prokhorov Metric and Error Bounds}\label{app:lp_details}

\begin{definition}[L\'evy-Prokhorov Metric]\label{def:levy_prokhorov}
Consider two probability measures $F$, $G$ over $\RR^d$, we define the L\'evy-Prokhorov distance between $F$ and $G$ relative to a norm $\norm{\cdot}$ on $\RR^d$ as
\begin{equation}
\pi(F, G) = \inf\{\epsilon > 0 : F(A) < G(A^\epsilon) + \epsilon \; \forall \; A\},
\end{equation}
where $A$ ranges over the Borel sets in $\RR^d$ and $A^\epsilon$ is an $\epsilon$-neighborhood of $A$ in the given norm.
\end{definition}

\begin{lemma}\label{lem:lp_symmetric}
$\pi(F, G) = \pi(G, F)$.
\end{lemma}

\begin{proof}
Suppose that for all Borel sets $A$, we have $F(A) < G(A^\epsilon) + \epsilon$. Then,
\begin{equation}
F(A^c) < G((A^c)^\epsilon) + \epsilon.
\end{equation}
For $\epsilon > 0$, define $A^{-\epsilon}$ to be the set of all points $x$ such that the open $\epsilon$-ball centered at $x$ is contained in $A$. Then, observe that
\begin{equation}
\ol{(A^c)^\epsilon} = \ol{(A^{-\epsilon})^c},
\end{equation}
which means that we can write
\begin{align}
1 - F(A) &< 1 - G(A^{-\epsilon}) + \epsilon \\
G(A^{-\epsilon}) &< F(A) + \epsilon
\end{align}
Now, since for all $A$, we have $A \subseteq (A^\epsilon)^{-\epsilon}$, so we can substitute
\begin{equation}
G(A) \leq G((A^\epsilon)^{-\epsilon})  < F(A^\epsilon) + \epsilon.
\end{equation}
Thus, $\pi(F, G) = \pi(G, F)$.
\end{proof}

\begin{definition}[Convolution]\label{def:convolution}
If $\varphi, \gamma$ are the probability densities of the distributions $F, G$, then the probability density of the convolution $F \ast G$ is defined to be
\begin{equation}
\rho(t) = \int_{\RR^d} \varphi(x) \gamma(t - x) \diff x.
\end{equation}
\end{definition}

\begin{definition}[Characteristic Function]\label{def:char_func}
The characteristic function of a distribution $F$ with density $\varphi$ is defined as
\begin{equation}
f(t) = \EE[e^{\iunit t X}] = \int_\RR e^{\iunit t x} \varphi(x) \diff x,
\end{equation}
where in the above $X \sim F$ is a random variable.
\end{definition}

\begin{lemma}\label{lem:berkes_philipp_helper}
Let $F$, $G$, and $H$ be probability distributions on $\RR^d$, and let $r > 0$. Then, the following inequality holds:
\begin{equation}
\pi(F, G) \leq \pi(F \ast H, G \ast H) + 2 \max\{r, H(\norm{x} \geq r)\},
\end{equation}
\end{lemma}

\begin{proof}
Let $X \sim F$, $Y \sim H$ be independent random variables. Let $A, K \subseteq \RR^d$. We know that
\begin{equation}
\begin{aligned}
& (F \ast H)(A) = \Pr[X + Y \in A] \geq \\
\geq{}& \Pr[X + k \in A \; \forall k \in K] \Pr[Y \in K] \geq \\
\geq{}& \Pr[X + k \in A \; \forall k \in K] + \Pr[Y \in K] - 1 = \\
={}& F(\{x : x + K \subseteq A\}) + H(K) - 1.
\end{aligned}
\end{equation}
Thus, letting $K$ be the open ball of radius $\epsilon$, we have that
\begin{equation}
(F \ast H)(A^\epsilon) \geq F(A) + H(\norm{x} < \epsilon) - 1 =
F(A) - H(\norm{x} \geq \epsilon).
\end{equation}
Note that we use the shorthand notation $H(\norm{x} < \epsilon)$ to mean $H(\{x : \norm{x} < \epsilon)\}$.

Now, we can also argue that
\begin{equation}
\begin{aligned}
& (F \ast H)(A) = \Pr[X + Y \in A] \leq \\
\leq{}& \Pr[X \in A - K]\Pr[Y \in K] + \Pr[Y \in K^c] \leq \\
\leq{}& \Pr[X \in A - K] + \Pr[Y \in K^c] =
F(A - K) + H(K^c).
\end{aligned}
\end{equation}
This gives us the inequality
\begin{equation}
(F \ast H)(A) \leq F(A^\epsilon) + H(\norm{x} \geq \epsilon).
\end{equation}
The same holds when replacing $F$ with $G$. Now, if we let $L = \pi(F \ast H, G \ast H)$, then we must have by the above and by the definition of the L\'evy-Prokhorov metric that
\begin{equation}
\begin{aligned}
& F(A) \leq (F \ast H)(A^r) + H(\norm{x} \geq r) \leq \\
\leq{}& (G \ast H)(A^{r + L}) + H(\norm{x} \geq r) + L \leq \\
\leq{}& G(A^{2r + L}) + 2H(\norm{x} \geq r) + L \leq \\
\leq{}& G(A^{L + 2\max\{r, H(\norm{x} \geq r)\}}) + 2H(\norm{x} \geq r) + L + 2\max\{r, H(\norm{x} \geq r)\},
\end{aligned}
\end{equation}
which proves that
\begin{equation}
\pi(F, G) \leq \pi(F \ast H, G \ast H) + 2 \max\{r, H(\norm{x} \geq r)\}.
\end{equation}
\end{proof}

\begin{lemma}\label{lem:berkes_philipp}
(This lemma is adapted from Lemma 2.2 of \citet{Berkes_1979}). Let $F$ and $G$ be probability distributions on $\RR^d$ with characteristic functions $f$ and $g$, respectively. Then, for some sufficiently large $T = O(d)$, we have that the L\'evy-Prokhorov metric relative to the Euclidean norm is bounded as
\begin{equation}
\pi_2(F, G) < \parens{\frac{T}{\pi}}^d \int_{\twonorm{u} \leq T} \abs{f(u) - g(u)} \diff u + F(\{x : \twonorm{x} \geq T/2\}) + 16 d \frac{\log T}{T}.
\end{equation}
\end{lemma}

\begin{proof}
Let $H$ be a distribution on $\RR^d$ with density $\nu(x)$ and characteristic function $h \in L^1$. Let $F_1 = F \ast H$ and $G_1 = G \ast H$. Let $\varphi$ and $\gamma$ be the probability densities of $F_1$ and $G_1$, respectively. Their characteristic functions are $f_1 = f h$ and $g_1 = g h$. Now, using the definition of convolution and applying the inverse Fourier transform, we have
\begin{equation}
\begin{aligned}
& \abs{\varphi(x) - \gamma(x)} = (2\pi)^{-d} \abs{\int_{\RR^d} e^{-\iunit \angles{u, x}} (f_1(u) - g_1(u)) \diff u} \leq \\
\leq{}& (2\pi)^{-d} \int_{\RR^d} \abs{(f(u) - g(u))}\abs{h(u)} \diff u \leq \\
\leq{}& (2\pi)^{-d} \parens{\int_{\twonorm{u} \leq T} \abs{(f(u) - g(u))} \diff u + 2 \int_{\twonorm{u} \geq T} \abs{h(u)} \diff u},
\end{aligned}
\end{equation}
where $T$ is any real number. Then, for any Borel set $B \in \RR^d$,
\begin{equation}
\begin{aligned}
& F_1(B) - G_1(B) \leq
F_1(B \cap \{\twonorm{x} \leq T\}) - G_1(B \cap \{\twonorm{x} \leq T\}) + F_1(\twonorm{x} \geq T) \leq \\
\leq{}& \int_{\twonorm{x} \leq T} \abs{\varphi(x) - \gamma(x)} \diff x + F(\twonorm{x} \geq T/2) + H(\twonorm{x} \geq T/2) \leq \\
\leq{}& \parens{\frac{T}{\pi}}^d \parens{\int_{\twonorm{u} \leq T} \abs{(f(u) - g(u))} \diff u + 2 \int_{\twonorm{u} \geq T} \abs{h(u)} \diff u} + \\
& + F(\twonorm{x} \geq T/2) + H(\twonorm{x} \geq T/2).
\end{aligned}
\end{equation}
Since for all $B$
\begin{equation}
F_1(B) = G_1(B) + (F_1(B) - G_1(B)) \leq G_1(B^{(F_1(B) - G_1(B))}) + (F_1(B) - G_1(B)),
\end{equation}
we have that
\begin{equation}
\pi(F_1, G_1) \leq \abs{F_1(B) - G_1(B)}.
\end{equation}
We also know from \Cref{lem:berkes_philipp_helper} that for any $r > 0$,
\begin{equation}
\pi(F, G) \leq \pi(F_1, G_1) + 2 \max\{r, H(\twonorm{x} \geq r)\}.
\end{equation}
Now, if we choose $r \leq T/2$, we can put this together to get
\begin{equation}
\begin{aligned}
\pi(F, G) \leq \parens{\frac{T}{\pi}}^d
\parens{\int_{\twonorm{u} \leq T} \abs{(f(u) - g(u))} \diff u + 2 \int_{\twonorm{u} \geq T} \abs{h(u)} \diff u} + \\
+ F(\twonorm{x} \geq T/2) + 3 \max\{r, H(\twonorm{x} \geq r)\}.
\end{aligned}
\end{equation}
Now, let $\sigma = 3 d^{1/2} T^{-1} \log^{1/2} T$ and $r = 5d T^{-1} \log T$. Let $H$ be a distribution with probability density
\begin{equation}
\nu(x) = (2\pi \sigma^2)^{-d/2} \exp\parens{-\frac{1}{2 \sigma^2} \sum_{j=1}^d x_j^2}.
\end{equation}
Then,
\begin{equation}
\begin{aligned}
& h(u) =
\int_{\RR^d} e^{\iunit \angles{u, x}} \nu(x) \diff x = \\
={}& (2\pi \sigma^2)^{-d/2} \prod_{j=1}^d \int_\RR e^{\iunit u_j x - \frac{1}{2 \sigma^2} x^2} \diff x = \\
={}& (2\pi \sigma^2)^{-d/2} \prod_{j=1}^d \int_\RR
e^{-\frac{1}{2 \sigma^2} (x^2 - 2\iunit \sigma^2 u_j x - \sigma^4 u_j^2 + \sigma^4 u_j^2)} \diff x = \\
={}& (2\pi \sigma^2)^{-d/2} \prod_{j=1}^d \int_\RR
e^{-\frac{1}{2 \sigma^2} (x - \iunit \sigma^2 u_j)^2 - \frac{1}{2}\sigma^2 u_j^2} \diff x = \\
={}& (2\pi \sigma^2)^{-d/2} \exp\parens{-\frac{1}{2} \sigma^2 \sum_{j=1}^d u_j^2} \parens{\int_\RR
e^{-\frac{1}{2 \sigma^2} x^2} \diff x}^d = \\
={}& \exp\parens{-\frac{1}{2} \sigma^2 \sum_{j=1}^d u_j^2}.
\end{aligned}
\end{equation}
Note that we are able to remove the $\iunit \sigma^2 u_j$ term in the exponent in the integrand, because the integrand is an entire function and so its integral around a rectangle with one side on the real axis and one side shifted by $-\iunit \sigma^2 u_j$ must be zero. Since the the contribution of the other two sides of the rectangle goes to zero, the value of the integral remains the same when moved to the real axis.

So now, we have that
\begin{equation}
\begin{aligned}
& \pi(F, G) \leq \\
\leq{} & \parens{\frac{T}{\pi}}^d \parens{\int_{\twonorm{u} \leq T} \abs{(f(u) - g(u))} \diff u + 2 \int_{\twonorm{u} \geq T} \exp\parens{-\frac{1}{2} \sigma^2 \sum_{j=1}^d u_j^2} \diff u} + \\
& + F(\twonorm{x} \geq T/2) + 3 \max\braces{r, (2\pi \sigma^2)^{-d/2} \int_{\twonorm{u} \geq r} \exp\parens{-\frac{1}{2 \sigma^2} \sum_{j=1}^d u_j^2} \diff u}.
\end{aligned}
\end{equation}

Now, observe that for any $A$,
\begin{equation}
\begin{aligned}
& \int_{\twonorm{u} \geq A} \exp\parens{-\frac{1}{2} \sum_{j=1}^d u_j^2} \diff u =
(2\pi)^{d/2} \Pr[\chi_d^2 \geq A^2] =
(2\pi)^{d/2} \Pr\brackets{e^{\frac{3}{8}\chi_d^2} \geq e^{\frac{3}{8} A^2}} \leq \\
\leq{}& (2\pi)^{d/2} e^{-3 A^2 / 8} \EE\brackets{e^{\frac{3}{8}\chi_d^2}} =
(2\pi)^{d/2} e^{-3 A^2 / 8} 2^d.
\end{aligned}
\end{equation}

Going back to our previous expression, we can substitute $v_j = \sigma u_j$ for the second integral and $v_j = u_j / \sigma$ for the third, which gives us
\begin{equation}
\begin{aligned}
& \pi(F, G) \leq \parens{\frac{T}{\pi}}^d \parens{\int_{\twonorm{u} \leq T} \abs{(f(u) - g(u))} \diff u + \frac{2}{\sigma^d} \int_{\abs{v} \geq \sigma T} \exp\parens{-\frac{1}{2} \sum_{j=1}^d v_j^2} \diff v} + \\
& \hspace{4em} + F(\twonorm{x} \geq T/2) + 3 \max\braces{r, (2\pi)^{-d/2} \int_{\abs{v} \geq r / \sigma} \exp\parens{-\frac{1}{2} \sum_{j=1}^d v_j^2} \diff v} = \\
={}& \parens{\frac{T}{\pi}}^d \parens{\int_{\twonorm{u} \leq T} \abs{(f(u) - g(u))} \diff u + \frac{2}{\sigma^d} (2\pi)^{d/2} e^{-3 T^2 \sigma^2 / 8} 2^d} + \\
& + F(\twonorm{x} \geq T/2) + 3 \max\braces{r, e^{-\frac{3 r^2}{8 \sigma^2}} 2^d}.
\end{aligned}
\end{equation}

Now, since
\begin{equation}
\frac{3 r^2}{8 \sigma^2} = \frac{3}{8} \frac{25 d^2 T^{-2} \log^2 T}{9 d T^{-2} \log T} =
\frac{25 d \log T}{24},
\end{equation}
so the second term in the maximum becomes
\begin{equation}
T^{-\frac{25 d}{24} 2^{d}}
\end{equation}
which becomes subleading to
\begin{equation}
r = \frac{5 d \log T}{T}
\end{equation}
for large $d$. Also, we have that
\begin{equation}
\frac{2}{\sigma^d} (2\pi)^{d/2} e^{-3 T^2 \sigma^2 / 8} 2^d =
\frac{2^{d+1} (2 \pi)^{d/2}}{\parens{3 d^{1/2} T^{-1} \log^{1/2} T}^d} T^{-\frac{27}{8}d},
\end{equation}
which becomes subleading if $T$ is sufficiently large. Increasing the coefficient for $r$ by one to strictly dominate the subleading terms, we arrive at the final expression:
\begin{equation}
\pi(F, G) \leq \parens{\frac{T}{\pi}}^d \int_{\twonorm{u} \leq T} \abs{(f(u) - g(u))} \diff u + F(\twonorm{x} \geq T/2) + \frac{16 d \log T}{T}.
\end{equation}
\end{proof}

\begin{corollary}\label{cor:modified_berkes_philipp}
Let $F$ and $G$ be two distributions on $\RR^d$ with density functions $\varphi$ and $\gamma$ and characteristic functions $f$ and $g$ respectively. Assume that there exists some $C$ such that
\begin{equation}
\int_{\norm{x}_\infty \geq C} f(x) \diff x \leq \mu.
\end{equation}
Then, there exists a universal constant $K$ such that for all $T \geq \max(2C \sqrt{d}, K d)$, we have that
\begin{equation}
\pi(F, G) \leq \parens{\frac{T}{\pi}}^d \int_{\norm{u}_\infty \leq T} \abs{(f(u) - g(u))} \diff u + \frac{16 d \log T}{T} + \mu.
\end{equation}
\end{corollary}

\begin{proof}
We need to convert from the Euclidean norm used in \Cref{lem:berkes_philipp} to the infinity norm. We know that $\norm{u}_\infty \leq \norm{u}_2 \leq \sqrt{d} \norm{u}_\infty$, so after performing that change, the claim immediately follows.
\end{proof}

\begin{lemma}\label{lem:lp_bound_from_moments}
(This has been adapted from Corollary 38 of \citet{Anschuetz_2025}). Let $F, G$ be distributions on $\RR^d$ with densities $\varphi, \gamma$ and characteristic functions $f, g$. Assume that each moment of order $k' \leq k$ of $F$ differs from that of $G$ by an additive error of at most $\epsilon > 0$ and assume that all moments of order $k' > k$ are bounded by $(C \sqrt{k'})^{k'}$ for some constant $C$. Let $\zeta > 0$ be a sufficiently small constant. Assume that
\begin{equation}
\int_{\norm{x}_\infty \geq \frac{1}{2 \sqrt{d}}\min\braces{\frac{1 - \zeta}{d} \log(\epsilon^{-1}), \frac{\sqrt{k+1}}{2e C d}}} f(x) \diff x \leq \mu.
\end{equation}
Also assume that
\begin{equation}
d^2 = o\parens{\min\braces{\frac{\log(\epsilon^{-1})}{\log \log(\epsilon^{-1})}, \frac{\sqrt{k}}{\log k}}}.
\end{equation}
Then, the L\'evy-Prokhorov distance is bounded by
\begin{equation}
\pi(F, G) = O\parens{\frac{d^2 \log \log (\epsilon^{-1})}{\log(\epsilon^{-1})} + \frac{d^2 \log k}{\sqrt{k}} + \mu}.
\end{equation}
\end{lemma}

\begin{proof}
Let
\begin{equation}
T = \min\braces{\frac{1 - \zeta}{d}\log(\epsilon^{-1}), \frac{\sqrt{k + 1}}{2 e C d}}.
\end{equation}
This, with the stated assumptions, satisfies the conditions for applying \Cref{cor:modified_berkes_philipp}, which gives us that
\begin{equation}
\pi(F, G) \leq \widetilde{O}(T^{2d}) \sup_{\norm{u}_\infty \leq T} \abs{f(u) - g(u)} + O\parens{\frac{d \log T}{T}} + O(\mu).
\end{equation}
Since
\begin{equation}
dT \leq \frac{\sqrt{k + 1}}{2 e C},
\end{equation}
we can apply Lemma 37 of \citet{Anschuetz_2025} to say that
\begin{equation}
\sup_{\norm{u}_\infty \leq T} \abs{f(u) - g(u)} \leq
\sup_{\norm{u}_1 \leq d T} \abs{f(u) - g(u)} \leq
\epsilon e^{d T} + 2^{1 - k} \leq \epsilon^\zeta + 2^{1 - k}.
\end{equation}
We also have that
\begin{equation}
T^{2d} = \exp(2d \log T) = \exp\parens{o(\min\{\log(\epsilon^{-1}), \sqrt{k}\})}.
\end{equation}
So,
\begin{equation}
\widetilde{O}(T^{2d}) \sup_{\norm{u}_\infty \leq T} \abs{f(u) - g(u)} + O\parens{\frac{d \log T}{T}} \leq O\parens{\epsilon^{\zeta + \sqrt{k}} + e^{\sqrt{k} + (1 - k) \log 2}}.
\end{equation}
Now, looking at the second term, we have that
\begin{equation}
\begin{aligned}
& O\parens{\frac{d \log T}{T}} \leq O\parens{d \min\braces{\log \log (\epsilon^{-1}), \log k} \max\braces{\frac{d}{\log(\epsilon^{-1})}, \frac{d}{\sqrt{k}}}} \leq \\
\leq{} & O\parens{\frac{d^2 \log \log (\epsilon^{-1})}{\log(\epsilon^{-1})} + \frac{d^2 \log k}{\sqrt{k}}}.
\end{aligned}
\end{equation}
This second term dominates the first, which means that
\begin{equation}
\pi(F, G) = O\parens{\frac{d^2 \log \log (\epsilon^{-1})}{\log(\epsilon^{-1})} + \frac{d^2 \log k}{\sqrt{k}} + \mu}.
\end{equation}
\end{proof}

\section{Numerical Experiments}\label{app:numerics}

In this appendix, we provide small-scale numerical experiments to illustrate our main results. We train our QNN model to perform fragment classification on the Temperley-Lieb model, using simple toy datasets. It is important to note that these experiments are not intended to serve as a demonstration of a direct practical application of our work, as the model we consider already admits efficient classical simulation techniques. They are also not meant to demonstrate any specific asymptotic scaling behavior proven in our main results, as simulating the QNN model classically becomes impractical for larger system sizes. The experiments serve as illustrative examples to demonstrate QNN trainability for the fragment classification task.

\subsection{QNN Training}

\begin{figure}[ht]
    \centering
    \includegraphics[width=0.75\linewidth]{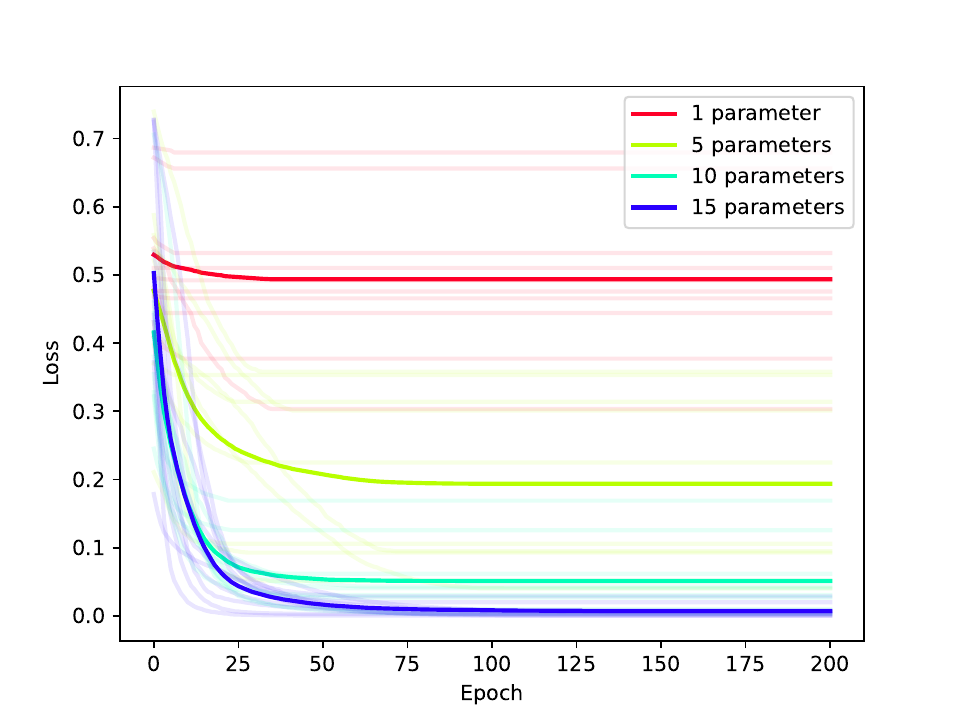}
    \caption{Training curves for QNN models with $1$, $5$, $10$, and $15$ parameters on the $4$-qubit Temperley-Lieb dataset. For each number of parameters, $10$ QNNs were randomly initialized and trained for $200$ epochs with a learning rate of $0.1$. Training was stopped if the loss was decaying slower than $5\%$ every $5$ epochs or if it reached less than $0.01$. The faded lines show the individual training curves, and the solid lines show their averages. For these diagrams, the loss value is adjusted to scale from $0$ to $1$ instead of from $-1$ to $0$.}
    \label{fig:training_curves_4qubits}
\end{figure}

\begin{figure}[ht]
    \centering
    \includegraphics[width=0.75\linewidth]{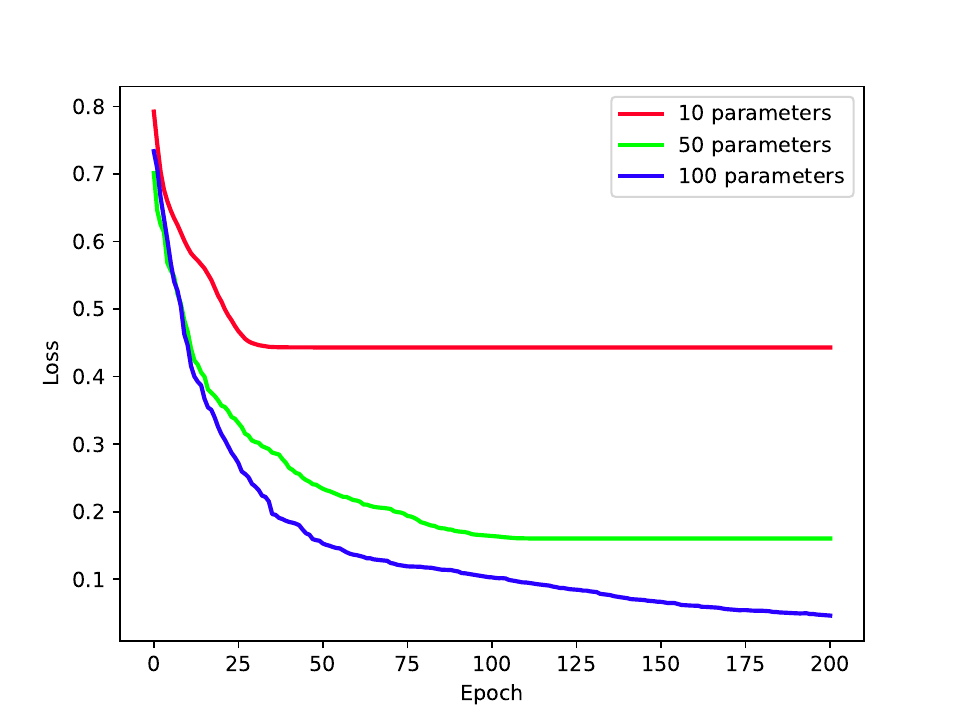}
    \caption{Training curves for QNN models with $10$, $50$, and $100$ parameters on the $6$-qubit Temperley-Lieb dataset. For each number of parameters, a QNN was randomly initialized and trained for $200$ epochs with a learning rate of $1$. Training was stopped if the loss was decaying slower than $5\%$ every $5$ epochs or if it reached less than $0.01$. Note that the increase in the number of parameters needed compared to \Cref{fig:training_curves_4qubits} is due not just to the larger system size, but also to the increased number of dataset entries, doubling the dimension $N_a$.}
    \label{fig:training_curves_6qubits}
\end{figure}

We consider a Schur basis for the Temperley-Lieb model on $2$-level quantum systems constructed in correspondence with the ``dots and dimers'' picture of \citet{Moudgalya_2022}. For simplicity, we take a subset of such a basis consisting of states that can be expressed as tensor products of computational basis states and the EPR state $\ket{\Phi^+}$. We demonstrate that our QNN model can be trained to distinguish states of this form with different configurations of entangled pairs, while generalizing across any valid state on the non-entangled qubits. We consider the following $2$-entry dataset on a $4$-qubit system:
\begin{equation}
\cD_4 = \braces{\ket{\Phi^+} \otimes \ket{01}, \ket{\Phi^+}^{\otimes 2}},
\end{equation}
and the following $4$-entry dataset on a $6$-qubit system:
\begin{equation}
\cD_6 = \braces{\ket{\Phi^+} \otimes \ket{0101}, \ket{0101} \otimes \ket{\Phi^+}, \ket{\Phi^+}^{\otimes 2} \otimes \ket{01}, \ket{\Phi^+}^{\otimes 3}}.
\end{equation}
We then construct the QNN ansatz using local generators of the form $\ket{\Phi^+}\bra{\Phi^+}_{(i,i+1)}$, randomly selecting the Hamiltonian, the coefficients $t_1, \dots, t_p, t', t''$ as in \Cref{sec:ansatz_const}. The QNN is initialized with random parameters $\bs{\theta}$ and optimized by gradient descent.

\begin{figure}[ht]
    \centering
    \includegraphics[width=0.75\linewidth]{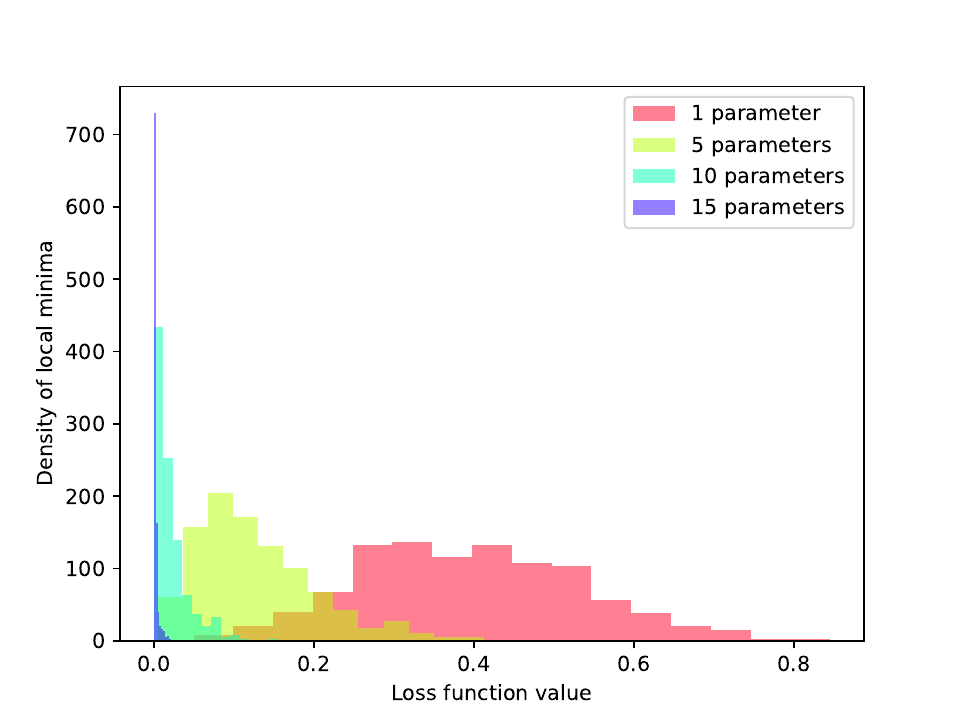}
    \caption{Distribution of local minima reached during training for the $4$-qubit Temperley-Lieb dataset. For this plot, $1000$ QNNs were randomly initialized for each number of parameters and trained for $200$ epochs with a learning rate of $0.1$. The data plotted on the $x$-axis is the final loss achieved during training. Since we can assume that during gradient descent the parameters converge to the ``closest'' local minimum, this is a proxy for the distribution of the local minima of the loss landscape. We can see that with $15$ parameters, the QNN enters the overparameterized regime since the distribution of minima becomes peaked at zero.}
    \label{fig:minima_distribution}
\end{figure}

The simulations were performed using the \texttt{qutip} library in Python. Since our simulations explicitly compute matrix products and exponentials of operators acting on the exponential-size Hilbert space, these methods become impractical for larger system sizes. That said, even at small system sizes our results---reported in \Cref{fig:training_curves_4qubits,fig:training_curves_6qubits,fig:minima_distribution}---support the correctness of the QNN construction and illustrate the phenomenon of the overparameterized regime: there exists a number of parameters for which the QNN will almost always reach the global optimum loss during training.

\subsection{Empirical Tests of the ETH Assumption}

Here we do a simple demonstration to support the claim that we can reasonably expect our Hamiltonians to satisfy the full ETH in each Krylov subspace for our analysis. For the Temperley-Lieb Hamiltonian considered in the QNN training experiments above, we numerically check that
\begin{equation}\label{eq:eth_numerics}
\mathbb{E}[\Tr(e^{-i H t} \ket{x} \bra{x} e^{i H t} O_x)^2] - \mathbb{E}[\Tr(e^{-i H t} \ket{x} \bra{x} e^{i H t} O_x)]^2 = O(1/N^*_\lambda),
\end{equation}
analogously to Equation (C3) of \citet{Fava_2025}, which we use in our Lemma 2. Here, $\ket{x}$ is a state from some subspace labelled by $\lambda$---in this case, we take a state of the form $\ket{\Phi^+} \otimes \ket{0101\dots}$, which is in a Krylov subspace whose dimension grows with system size. Note that for the purposes of our analysis, it is sufficient that quantities of this form decay with the system size (not necessarily as $O(1/N^*_\lambda)$).

\begin{figure}
    \centering
    \includegraphics[width=0.5\linewidth]{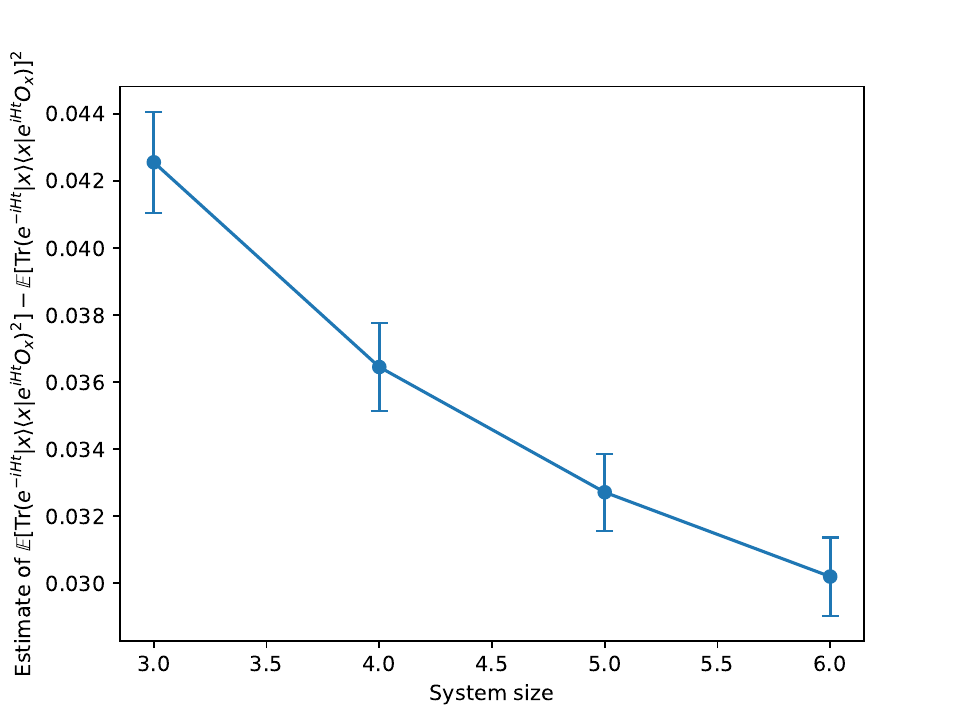}
    \caption{For each system size from $3$ to $6$, we sample $1000$ different instances of the Temperley-Lieb Hamiltonian and time-evolve them for $100$ different random times to estimate the quantity in \Cref{eq:eth_numerics}.}
    \label{fig:eth_numerics}
\end{figure}

\end{document}